\documentclass[pdftex]{elsarticle}
\usepackage{tikz}
\usepackage{amsmath}
\usepackage{amssymb}
\usepackage{amsthm}

\newtheorem{theorem}{Theorem}
\newtheorem{proposition}{Proposition}
\newtheorem{lemma}{Lemma}
\newtheorem{corollary}{Corollary}
\newtheorem{question}{Question}
\newtheorem{rem}{Remark}
\newtheorem{ex}{Example}

\newenvironment{remark}{\begin{rem}\begin{em}}{\end{em}\end{rem}}
\newenvironment{example}{\begin{ex}\begin{em}}{\end{em}\end{ex}}

\def\a{\alpha}
\def\f{\tilde f}
\def\h{\tilde h}
\def\s{\sigma}
\def\o{\oplus}
\def\d#1{|\hspace{-0.5mm}|#1|\hspace{-0.5mm}|}
\def\B{\mathbb{B}} 
\def\F{\mathcal{F}}
\def\C{\mathcal{C}}
\def\P{\mathcal{P}}
\def\ONE{\mathbf{1}}
\def\One{\text{\sc one}}
\def\Zero{\text{\sc zero}}
\def\Sub{\text{\sc{sub}}}
\def\esd{\F_{\text{\sc{esd}}}}
\def\osd{\F_{\text{\sc{osd}}}}
\def\eosd{\F_{\text{\sc{eosd}}}}
\begin{document}
%%%%%%%%%%%%%%%%%%%%%%%%%%%%%%%%%%%%%%%%%%%%%%%%%%%%%%%%%%%%%%%%%%%%%%
%%%%%%%%%%%%%%%%%%%%%%%%%%%%%%%%%%%%%%%%%%%%%%%%%%%%%%%%%%%%%%%%%%%%%%

\begin{frontmatter}

\title{Fixed point theorems for Boolean networks\\ expressed in terms of forbidden subnetworks}
\tnotetext[]{A preliminary version of this paper \cite{R2011} was presented at the 17th International Workshop on Cellular Automata and Discrete Complex Systems (AUTOMATA 2011).}

\author{Adrien Richard}
\address{Laboratoire I3S, UMR7271 CNRS \& University of Nice-Sophia Antipolis\\  2000, route des lucioles Les Algorithmes - b\^at. Euclide B, 
06900 Sophia Antipolis, France}
\cortext[]{Corresponding author. Phone: +33 4 92 94 27 51. Email: {\tt richard@unice.fr}}

\begin{abstract}
We are interested in fixed points in Boolean networks, {\em i.e.} functions $f$ from $\{0,1\}^n$ to itself. We define the subnetworks of $f$ as the restrictions of $f$ to the subcubes of  $\{0,1\}^n$, and we characterizes a class $\F$ of Boolean networks satisfying the following property: Every subnetwork of $f$ has a unique fixed point if and only if $f$ has no subnetwork in $\F$. This characterization generalizes the fixed point theorem of Shih and Dong, which asserts that if for every $x$ in $\{0,1\}^n$ there is no directed cycle in the directed graph whose the adjacency matrix is the discrete Jacobian matrix of $f$ evaluated at point $x$, then $f$ has a unique fixed point. Then, denoting by $\C^+$ (resp. $\C^-$) the networks whose the interaction graph is a positive (resp. negative) cycle, we show that the non-expansive networks of $\F$ are exactly the networks of $\C^+\cup \C^-$; and for the class of non-expansive networks we get a ``dichotomization'' of the previous forbidden subnetwork theorem: Every subnetwork of $f$ has at most (resp. at least) one fixed point if and only if $f$ has no subnetworks in $\C^+$ (resp. $\C^-$) subnetwork. Finally, we prove that if $f$ is a conjunctive network then every subnetwork of $f$ has at most one fixed point if and only if $f$ has no subnetworks in $\C^+$.
\end{abstract}


\end{frontmatter}

%%%%%%%%%%%%%%%%%%%%%%%%%%%%%%%%%%%%%%%%%%%%%%%%%%%%%%%%%%%%%%%%%%%%%%
\section{Introduction}\label{sec:in}
%%%%%%%%%%%%%%%%%%%%%%%%%%%%%%%%%%%%%%%%%%%%%%%%%%%%%%%%%%%%%%%%%%%%%%

A function $f$ from $\{0,1\}^n$ to itself is often seen as a Boolean network with $n$ components. On one hand, the dynamics of the network is described by the iterations of $f$; for instance, with the synchronous iteration scheme, the dynamics is described by the recurrence $x^{t+1}=f(x^t)$. On the other hand, the ``structure'' of the network is described by a directed graph $G(f)$: The vertices are the $n$ components, and there exists an arc from $j$ to $i$ when the evolution of the $i$th component depends on the evolution of the $j$th~one.

\medskip
Boolean networks have many applications. In particular, from the seminal works of Kauffman {\cite{Kauffman69}} and Thomas {\cite{Thomas73}}, they are extensively used to model gene networks. In most cases, fixed points are of special interest. For instance, in the context of gene networks, they correspond to stable patterns of gene expression at the basis of particular biological processes.

\medskip
Importance of fixed point leads researchers to find conditions for the existence and the uniqueness of a fixed point. Such a condition was first obtained by Robert {\cite{R80}}, who proved that {\emph{if $G(f)$ has no directed cycle, then $f$ has a unique fixed point}}. This result was then generalized by Shih and Dong {\cite{SD05}}. They associated to each point $x$ in $\{0,1\}^n$ a local interaction graph $Gf(x)$, which is a subgraph of $G(f)$ defined as the directed graph whose the adjacency matrix is the discrete Jacobian matrix of $f$ evaluated at point~$x$, and they proved that {\emph{if $Gf(x)$ has no directed cycle for all $x$ in $\{0,1\}^n$, then $f$ has a unique fixed point}}. 

\medskip
In this paper, we generalize Shih-Dong's theorem using, as main tool, the {\em subnetworks} of $f$, that is, the networks obtained from $f$ by fixing to $0$ or $1$ some components. The organization is the following. 
After introducing the main concepts in Section~\ref{sec:def}, we formally state some classical results connected to this work, as Robert's and Shih-Dong's theorems. 
In Section~\ref{sec:main}, we define the class $\F$ of even and odd-self-dual networks, and we prove the main result of this paper, the following characterization: {\emph{$f$ and all its subnetworks have a unique fixed point if and only if $f$ has no subnetworks in~$\F$}}. The rest of the paper discusses this ``forbidden subnetworks theorem''. 
In section~\ref{sec:shih}, we show that it generalizes Shih-Dong's theorem. More precisely, we show how it can be used to replace the condition ``$Gf(x)$ has no cycles for all $x$'' in Shih-Dong's theorem by a weaker condition of the form ``$Gf(x)$ has short cycles for few points $x$''. 
In section~\ref{sec:asyn}, we study the effect of the absence of subnetwork in $\F$ on the asynchronous state graph of $f$ (which is a directed graph on $\{0,1\}^n$ constructed from the asynchronous iterations of $f$ and proposed by Thomas {\cite{Thomas73}} as a model for the dynamics of gene networks). 
Section~\ref{sec:sub} gives some reflexions on the characterization of properties by forbidden subnetworks. In particular, it is showed that there is not a lot of properties that are interesting to characterize in terms of forbidden subnetworks. 
In Section~\ref{sec:cycle}, we compare $\F$ with the with the classes $\C^+$ (resp. $\C^-$) of networks $f$ such that the interaction graph $G(f)$ is a positive (resp. negative) cycle. We show that $\C^+$ (resp. $\C^-$) contains exactly the non-expansive even-self-dual (resp. odd-self-dual) networks, in such a way that $\C^+\cup\C^-$ equals the non-expansive networks of $\F$. 
This result is used in Section~\ref{sec:non-exp} to obtain a strong version of the main result for non-expansive networks: {\em If $f$ is non-expansive, then $f$ and all its subnetworks have at least (resp. at most) one fixed point if and only if $f$ has no subnetworks in $\C^-$ (resp. $\C^+$)}. 
In Section~\ref{sec:and}, we focus on conjunctive networks. We prove that {\em if $f$ is a conjunctive network, then $f$ and all its subnetworks have at most one (resp. a unique) fixed point if and only if $f$ has no subnetworks in $\C^+$ (resp. $\C^+\cup\C^-$)}. Finally, we show that, for conjunctive networks, the absence of subnetwork in $\C^{\pm}$ can be easily verified from the chordless cycles of $G(f)$. 

%%%%%%%%%%%%%%%%%%%%%%%%%%%%%%%%%%%%%%%%%%%%%%%%%%%%%%%%%%%%%%%%%%%%%%
%%%%%%%%%%%%%%%%%%%%%%%%%%%%%%%%%%%%%%%%%%%%%%%%%%%%%%%%%%%%%%%%%%%%%%
\section{Preliminaries}\label{sec:def}
%%%%%%%%%%%%%%%%%%%%%%%%%%%%%%%%%%%%%%%%%%%%%%%%%%%%%%%%%%%%%%%%%%%%%%
%%%%%%%%%%%%%%%%%%%%%%%%%%%%%%%%%%%%%%%%%%%%%%%%%%%%%%%%%%%%%%%%%%%%%%

%%%%%%%%%%%%%%%%%%%%%%%%%%%%%%%%%%%%%%%%%%%%%%%%%%%%%%%%%%%%%%%%%%%%%%
\subsection{Notations on hypercube}
%%%%%%%%%%%%%%%%%%%%%%%%%%%%%%%%%%%%%%%%%%%%%%%%%%%%%%%%%%%%%%%%%%%%%%

If $A$ and $B$ are two sets, then $A^B$ denotes the set of functions from $B$ to~$A$. Let $\B=\{0,1\}$ and let $V$ be a finite set. Elements of $\B^V$ are seen as {\em points} of the $|V|$-dimensional Boolean space, and the elements of $V$ as the {\em components} (or dimensions) of this space.  Given a point $x\in\B^V$ and a component $i\in V$, the image of $i$ by $x$ (the $i$-component of $x$) is denoted $x_i$ or $(x)_i$. The set of components $i$ such that $x_i=1$ is denoted $\ONE(x)$. For all $I\subseteq V$, we denote by $e_I$ the point of $\B^V$ such that $\ONE(e_I)=I$. Points $e_\emptyset$ and $e_V$ are often denoted $0$ and $1$, and we write $e_i$ instead of~$e_{\{i\}}$. Hence, $e_i$ may be seen as the base vector of $\B^V$ associated with dimension $i$. For all $x\in\B^V$, we set $\d{x}=|\ONE(x)|$. A point $x$ is said to be even (resp. odd) if $\d{x}$ is even (resp. odd). The sum modulo two is denoted $\o$. If $x$ and $y$ are two points of $\B^V$, then $x\o y$ is the point of $\B^V$ such that $(x\o y)_i=x_i\o y_i$ for all $i\in V$. The {\em Hamming distance} between $x$ and $y$ is $d(x,y)=\d{x\o y}$. Thus $d(x,y)$ is the number of components $i$ such that $x_i\neq y_i$. In this way $\d{x}=d(0,x)$. For all $I\subseteq V$ and $x\in\B^V$, the restriction of $x$ to $I$ is denoted $x|_I$, and the restriction of $x$ to $V\setminus I$ is denoted $x_{-I}$. If $i\in V$, we write $x|_i$ and $x_{-i}$ instead of $x|_{\{i\}}$ and $x_{-\{i\}}$. Also, if $\a\in\B$ then $x^{i\a}$ denotes the point  of $\B^V$ such that $(x^{i\a})_i=\a$ and $(x^{i\a})_{-i}=x_{-i}$.

%%%%%%%%%%%%%%%%%%%%%%%%%%%%%%%%%%%%%%%%%%%%%%%%%%%%%%%%%%%%%%%%%%%%%%
\subsection{Networks and subnetworks}
%%%%%%%%%%%%%%%%%%%%%%%%%%%%%%%%%%%%%%%%%%%%%%%%%%%%%%%%%%%%%%%%%%%%%%   

A {\bf (Boolean) network} on $V$ is a function $f:\B^V\to\B^V$. The elements of $V$ are the {\em components} or {\em automata} of the network, and $\B^V$ is the set of possible {\em states} or {\em configurations} for the network. At a given configuration $x\in\B^V$, the state of component $i$ is given by $x_i$. The {\em local transition function} associated with component $i$ is the function $f_i$ from $\B^V$ to $\B$ defined by $f_i(x)=f(x)_i$ for all $x\in\B^V$. Throughout this article, $f$ denotes a network on $V$.   

\medskip
We say that $f$ is {\bf non-expansive} if 
\[
\forall x,y\in\B^V,\qquad d(f(x),f(y))\leq d(x,y).
\]

\medskip
The {\bf conjugate} of $f$ is the network $\f$ on $V$ defined by  
\[
\forall x\in\B^V,\qquad \f(x)=f(x)\o x.  
\]

\medskip
Let $I$ be a non-empty subset of $V$ and $z\in\B^{V\setminus I}$. The {\bf subnetwork} of $f$ induced by $z$ is the network $h$ on $I$ defined by 
\[
\forall x\in\B^V\text{ with }x_{-I}=z,\qquad h(x|_I)=f(x)|_I.
\]
The subnetwork of $f$ induced by $z$ is thus the network obtained from $f$ by fixing to $z_i$ each component $i\in V\setminus I$. It can also be seen as the projection of the restriction of $f$ to the hyperplane defined by the equations ``$x_i=z_i$'', $i\in V\setminus I$. Note that, by definition, {\em $f$ is a subnetwork of itself}. A subnetwork of $f$ distinct from $f$ is a {\bf strict subnetwork}. 
Let $i\in V$, $\alpha\in\B$ and let $z\in\B^V$ with $z_i=\alpha$. The subnetwork of $f$ induced by $z|_i$ is denoted~$f^{i\alpha}$ and called {\bf immediate subnetwork} of $f$ induced by the hyperplane ``$x_i=\a$''. In other words, 
\[
\forall x\in\B^V,\qquad f^{i\alpha}(x_{-i})=f(x^{i\a})_{-i}.
\]
%It is easy to see that the conjugate of $f^{i\a}$ is equals to $\f^{i\a}$, that is, to the subnetwork of $\f$ induced by the hyperplane ``$x_i=\a$''.
%\begin{multline*}
%\widetilde{f^{i\a}}(x_{-i})
%=x_{-i}\o f^{i\a}(x_{-i})\\[1mm]
%%=x_{-i}\o f(x^{i\a})_{-i}
%=(x\o f(x^{i\a}))_{-i}\\[1mm]
%=\f(x^{i\a})_{-i}
%=\f^{i\a}(x_{-i}).
%\end{multline*}

%%%%%%%%%%%%%%%%%%%%%%%%%%%%%%%%%%%%%%%%%%%%%%%%%%%%%%%%%%%%%%%%%%%%%%
\subsection{Asynchronous state graph}
%%%%%%%%%%%%%%%%%%%%%%%%%%%%%%%%%%%%%%%%%%%%%%%%%%%%%%%%%%%%%%%%%%%%%%   

The {\bf asynchronous state graph} of $f$, denoted $\Gamma(f)$, is the directed graph with vertex set $\B^V$ and the following set of arcs:
\[
\{x\to x\o e_i\,|\,x\in\B^V,~i\in V,~f_i(x)\neq x_i\}
\]

\begin{remark}
Our interest for $\Gamma(f)$ lies in the fact that this state graph has been proposed by Thomas {\cite{Thomas73}} as a model for the dynamics of gene networks; see also {\cite{TA90}}. In this context, network components correspond to genes. At a given state $x$, the protein encoded by gene $i$ is ``present'' if $x_i=1$ and ``absent'' if $x_i=0$. The gene $i$ is ``on'' (transcripted) if $f_i(x)=1$ and ``off'' (not transcripted) if $f_i(x)=0$. And given an initial configuration $x$, the possible evolutions of the system are described by the set of paths of $\Gamma(f)$ starting from $x$. 
\end{remark}

The terminal strongly connected components of $\Gamma(f)$ are called {\bf attractors}. An attractor is {\bf cyclic} if it contains at least two points, and it is {\bf punctual} otherwise. Hence, $\{x\}$ is a punctual attractor of $\Gamma(f)$ if and only if $x$ is a fixed point of $f$, so both concepts are identical. 

%%%%%%%%%%%%%%%%%%
\begin{proposition}\label{pro:subdynamics} 
Let $I$ be non-empty subset of $V$ and let $h$ be the subnetwork of $f$ induced by some point $z\in\B^{V\setminus I}$. The asynchronous state graph of $h$ is isomorphic to the asynchronous state graph of $f$ induced by the set of points $x\in\B^V$ such that $x_{-I}=z$ (the isomorphism is $x\mapsto x_{-I}$).   
\end{proposition}
%%%%%%%%%%%%%%%%%%

%%%%%%%%%%%%%%%%%%
\begin{proof}
For all $x,y\in \B^V$ with $x_{-I}=y_{-I}=z$, and for all $i\in I$, we have $y=x\o e_i$ if and only if $x|_I=y|_I\o e_i$, thus $x\to y$ is an arc of $\Gamma(f)$ if and only if $x_{-I}\to y_{-I}$ is an arc of $\Gamma(h)$.
\end{proof}
%%%%%%%%%%%%%%%%%%

%%%%%%%%%%%%%%%%%%%%%%%%%%%%%%%%%%%%%%%%%%%%%%%%%%%%%%%%%%%%%%%%%%%%%%
\subsection{Criticality}
%%%%%%%%%%%%%%%%%%%%%%%%%%%%%%%%%%%%%%%%%%%%%%%%%%%%%%%%%%%%%%%%%%%%%%   

We say that $f$ is {\bf critical} for a property $\P$, if $f$ has the property $\P$ but no strict subnetworks of $f$ have this property. Let $\P_2$ be the property ``to have at least two fixed points'', and let $\P_0$ be the property ``to have no fixed point''. We say that $f$ is {\bf 2-critical} if $f$ is critical for the property $\P_2$, and we say that $f$ is {\bf 0-critical} if $f$ is critical for the property $\P_0$. Clearly, if $f$ is 2-critical, then there exists $x\in\B^V$ such that $x$ and $x\o 1$ are fixed points, and $f$ has no other fixed point (because if $x$ and $y$ are two fixed points and $x_i=y_i=\a$ then $x_{-i}$ and $y_{-i}$ are fixed points of $f^{i\a}$).

%%%%%%%%%%%%%%%%%%
\begin{proposition}\label{pro:critical_1} Let $f$ be a network on $V$. 
\begin{enumerate}
\item
If the asynchronous state graph of $f$ has multiple attractors, then $f$ has a 2-critical subnetwork.   
\item
If $f$ is non-expansive and if the asynchronous state graph of $f$ has a cyclic attractor, then $f$ has no fixed point and thus has a 0-critical subnetwork.   
\end{enumerate}
\end{proposition}
%%%%%%%%%%%%%%%%%%

%%%%%%%%%%%%%%%%%%
\begin{proof}
Suppose that $\Gamma(f)$ has two distinct attractors $X,Y\subseteq \B^V$. Let $x\in X$ and $y\in Y$ be such that $d(x,y)$ is minimal. Let $I=\ONE(x\o y)$ so that $x_{-I}=y_{-I}=z$. Let $h$ be the subnetwork of $f$ induced by $z$. Suppose that $x|_I$ is not a fixed point of $h$. Then, there exists $i\in I$ with $x_i\neq h_i(x|_I)=f_i(x)$. Thus $\Gamma(f)$ has an arc $x\to x\o e_i$ and $x\o e_i\in X$ because $x\in X$. Since $x_i\neq y_i$, we have $d(x\o e_i,y)<d(x,y)$, a contradiction. Thus $x|_I$ is a fixed point of $h$, and we prove with similar arguments that $y|_I$ is a fixed point of $h$. Thus $h$ has  multiple fixed points. Thus $h$ has necessarily a 2-critical subnetwork $g$, and since $g$ is a subnetwork of $f$ the first point is proved.

\medskip
For the second point, suppose in addition that $f$ is non-expansive, that $Y$ is a cyclic attractor ({\em i.e.} $|Y|>1$) and that $X$ is punctual {\em i.e.} reduces to a fixed point $x$ of $f$. Since $y|_I$ is a fixed point of $h$, we have $y|_I=f(y)|_I$ and using the fact that $f$ is non expansive we get
\begin{multline*}
d(f(x),f(y))=d(x,f(y))=
d(x|_I,f(y)|_I)+d(x_{-I},f(y)_{-I})=\\[1mm]
d(x|_I,y|_I)+d(y_{-I},f(y)_{-I})=|I|+d(y_{-I},f(y)_{-I})\leq d(x,y)=|I|.
\end{multline*}
Thus $d(y_{-I},f(y)_{-I})=0$ so $y_{-I}=f(y)_{-I}$. Consequently, $y$ is a fixed point of $f$, and $Y$ cannot be cyclic, a contradiction. Consequently, if $f$ is non-expansive and if $Y$ is a cyclic attractor, then $X$ is also a cyclic attractor, so $f$ has no fixed point and thus it has necessarily a 0-critical subnetwork.
\end{proof}
%%%%%%%%%%%%%%%%%%

%%%%%%%%%%%%%%%%%%%%%%%%%%%%%%%%%%%%%%%%%%%%%%%%%%%%%%%%%%%%%%%%%%%%%%
\subsection{Interaction graphs}
%%%%%%%%%%%%%%%%%%%%%%%%%%%%%%%%%%%%%%%%%%%%%%%%%%%%%%%%%%%%%%%%%%%%%%   

Notions and notations concerning digraphs are consistent with \cite{BJG09}. In particular, cycles and paths are seen as digraphs and thus have no repeated vertices. A {\bf signed digraph} $G=(V,A)$ consists in a set of vertices $V$ and a set of (signed) arcs $A\subseteq V\times\{-1,1\}\times V$. An arc $(i,s,j)\in A$ is an arc from $i$ to $j$ of sign~$s$. We say that $G$ is {\bf simple} if for every vertices $i,j\in V$ there is at most one arc from $i$ to $j$. The (unsigned) digraph obtained by forgetting signs is denoted $|G|$: The vertex set of $|G|$ is $V$ and the arc set of $|G|$ is the set of couples $(i,j)$ such that $G$ has at least one arc from $i$ to~$j$. A signed digraph $G'=(V',A')$ is a subgraph of $G$ (notation $G'\subseteq G$) if $V'\subseteq V$ and $A'\subseteq A$. A {\bf cycle} of $G$ is a simple subgraph $C$ of $G$ such that $|C|$ is a directed cycle. A {\bf positive} (resp. {\bf negative}) {\bf cycle} of $G$ is a cycle of $G$ with an even (resp. odd) number of negative arcs. A cycle of $C$ of $G$ is {\bf chordless} if $|C|$ is an induced subgraph of $|G|$ ({\em i.e.} $|C|$ can be obtained from $|G|$ be removing vertices only). 

\medskip
Let $f$ be a network on $V$ and two components $i,j\in V$. The {\bf discrete derivative} of $f_i$ with respect to $j$ is the function $f_{ij}$ from $\B^V$ to $\{-1,0,1\}$ defined by 
\[
\forall x\in\B^V,\qquad f_{ij}(x)=f_i(x^{j1})-f_i(x^{j0}).
\]
Discrete derivatives are usually stored under the form of a matrix, the {\em Jacobian matrix}. However, for our purpose, it is more convenient to store them under the form of a signed digraph.  

\medskip
For all $x\in\B^V$, we call {\bf local interaction graph} of $f$ evaluated at point $x$, and we denote by $Gf(x)$, the signed digraph with vertex set $V$ such that, for all $i,j\in V$, there is a positive (resp. negative) arc from $j$ to $i$ if $f_{ij}(x)$ positive (resp. negative). Note that $Gf(x)$ is simple. The {\bf (global) interaction graph} of $f$ is the signed digraph denoted by $G(f)$ and defined by: The vertex set is $V$ and, for all vertices $i,j\in V$, there is a positive (resp. negative) arc from $j$ to $i$ if $f_{ij}(x)$ is positive (resp. negative) for at least one $x\in\B^V$. Thus each local interaction graph $Gf(x)$ is a subgraph of the global interaction graph $G(f)$. More precisely, $G(f)$ is obtained by taking the union of all the $Gf(x)$. 

%%%%%%%%%%%%%%%%%%
\begin{proposition}\label{pro:subgraph}
Let $I$ be non-empty subset of $V$ and let $h$ be the subnetwork of $f$ induced by some point $z\in\B^{V\setminus I}$, and let $x\in\B^V$ with $x_{-I}=z$. Then:
\begin{enumerate}
\item
$Gh(x|_I)$ is an induced subgraph of $Gf(x)$;
\item
$G(h)$ is a subgraph of $G(f)$. 
\end{enumerate}
\end{proposition}
%%%%%%%%%%%%%%%%%%

%%%%%%%%%%%%%%%%%%
\begin{proof}
If $x\in\B^V$ and $x_{-I}=z$, then for all $i,j\in I$,   
\[
h_{ij}(x|_I)=h_i(x^{j1}|_I)-h_i(x^{j0}|_I)=f_i(x^{j1})-f_i(x^{j0})=f_{ij}(x).
\]
This proves {\em 1.} and {\em 2.} is an obvious consequence.
\end{proof}
%%%%%%%%%%%%%%%%%%

%%%%%%%%%%%%%%%%%%%%%%%%%%%%%%%%%%%%%%%%%%%%%%%%%%%%%%%%%%%%%%%%%%%%%%
%%%%%%%%%%%%%%%%%%%%%%%%%%%%%%%%%%%%%%%%%%%%%%%%%%%%%%%%%%%%%%%%%%%%%%
\section{Some fixed point theorems}\label{sec:classical}
%%%%%%%%%%%%%%%%%%%%%%%%%%%%%%%%%%%%%%%%%%%%%%%%%%%%%%%%%%%%%%%%%%%%%%
%%%%%%%%%%%%%%%%%%%%%%%%%%%%%%%%%%%%%%%%%%%%%%%%%%%%%%%%%%%%%%%%%%%%%%

Robert proved in 1980 the following fundamental fixed point theorem \cite{R80,R86}. A short proof is given in \ref{appendix} (this proof uses an induction on subnetworks, a technic used in almost all proofs of this paper). 

%%%%%%%%%%%%%%%%%%
\begin{theorem}[Robert 1980]\label{thm:robert}
If $G(f)$ has no cycle then $f$ has a unique fixed point.
\end{theorem}
%%%%%%%%%%%%%%%%%%

\noindent
Robert also proved, in his french book~\cite{R95}, that if $G(f)$ has no cycle, then $\Gamma(f)$ has no cycle, so that every path of $\Gamma(f)$ leads to the unique fixed point of $f$ (strong convergence toward a unique fixed point).  

\medskip
The following theorem, proved by Aracena~\cite{A08} (see also \cite{ADG04}) in a slightly different setting, gives other very fundamental relationships between the interaction graph of $f$ and its fixed points. 

%%%%%%%%%%%%%%%%%%
\begin{theorem}[Aracena 2008]\label{thm:ara}
Suppose that $G(f)$ is strongly connected (and contains at least one arc).
\begin{enumerate}
\item
If $G(f)$ has no negative cycle then $f$ has at least two fixed points.
\item
If $G(f)$ has no positive cycle then $f$ has no fixed point.
\end{enumerate}
\end{theorem}
%%%%%%%%%%%%%%%%%%

The following theorem can be deduce from Aracena theorem with an induction on strongly connected components of $G(f)$, see \ref{appendix}. It gives a nice ``proof by dichotomy'' of Robert' theorem: The {\em existence} of a fixed point is established under the absence of {\em negative cycle} while the {\em unicity} under the absence of {\em positive cycle}.   

%%%%%%%%%%%%%%%%%%
\begin{theorem}\label{thm:thomas}~
\begin{enumerate}
\item
If $G(f)$ has no positive cycle then $f$ has at most one fixed point.
\item
If $G(f)$ has no negative cycle then $f$ has at least one fixed point.
\end{enumerate}
\end{theorem}
%%%%%%%%%%%%%%%%%%

\noindent
First point of Theorem~\ref{thm:thomas} can be seen as a Boolean version of first Thomas' rule, which asserts that the presence of a positive cycles in the interaction graph of a dynamical system is a necessary conditions for the presence of multiple stable states~\cite{T80} (see also \cite{KST07} and the references therein). 

\medskip
Second Thomas' rule asserts that the presence of a negative cycle is a necessary condition for the presence of cyclic attractors~\cite{T80,KST07}. Hence, the next theorem, proved in \cite{R10}, can be see as a Boolean version of second Thomas' rule.     

%%%%%%%%%%%%%%%%%%
\begin{theorem}[Richard 2010]
If $G(f)$ has no negative cycle, then $\Gamma(f)$ has no cyclic attractors.
\end{theorem}
%%%%%%%%%%%%%%%%%%

\noindent
Note that this theorem generalizes the second point of Theorem~\ref{thm:thomas}: If $\Gamma(f)$ has no cyclic attractor, then all the attractors are fixed points, and since there always exists at least one attractor, $f$ has at least one fixed point.

\medskip
The next theorem is a ``local version'' of Robert's theorem. It has been conjectured and presented as a combinatorial analog of the Jacobian conjecture in {\cite{SH99}}. It has be proved by Shih and Dong in  {\cite{SD05}}. 

%%%%%%%%%%%%%%%%%%
\begin{theorem}[Shih and Dong 2005]\label{thm:shihdong}
If $Gf(x)$ has no cycle for all $x\in\B^V$, then $f$ has a unique fixe point. 
\end{theorem}
%%%%%%%%%%%%%%%%%%

\noindent
This theorem generalizes Robert's one: If $G(f)$ has no cycle, then it is clear that each local interaction graph $Gf(x)$ has no cycle (because $Gf(x)\subseteq G(f)$). The original proof of Shih and Dong is quite involved. A much more simple proof is given in \ref{appendix}. 

\medskip
In a similar way, Remy, Ruet and Thieffry~\cite{RRT08} proved a local version of the first point of Theorem~\ref{thm:thomas}. They thus got the uniqueness part of Shih-Dong's theorem under weaker conditions.  

%%%%%%%%%%%%%%%%%%
\begin{theorem}[Remy, Ruet and Thieffry 2008]\label{thm:remy}
If $Gf(x)$ has no positive cycle for all $x\in\B^V$, then $f$ has at most one fixed point.  
\end{theorem}
%%%%%%%%%%%%%%%%%%

In view of the previous theorem, it very natural think about a local version of the second point of Theorem~\ref{thm:thomas}. 
 
%%%%%%%%%%%%%%%%%%
\begin{question}
Is it true that if $Gf(x)$ has no negative cycle for all $x\in\B^V$, then $f$ has at least one fixed point?  
\end{question}
%%%%%%%%%%%%%%%%%%

The following theorem, proved in \cite{R11}, only gives a very partial answer to this question (see \cite{RR2012} for another very partial answer).  

%%%%%%%%%%%%%%%%%%
\begin{theorem}[Richard 2011]\label{thm:non-exp}
If $f$ is non-expansive and if $Gf(x)$ has no negative cycle for all $x\in\B^V$, then $f$ has at least one fixed point.  
\end{theorem}
%%%%%%%%%%%%%%%%%%

%%%%%%%%%%%%%%%%%%
\begin{remark}
In all the theorems, Aracena one excepted, if the conditions are satisfied by $f$ then they are also satisfied by every subnetwork of $f$, in such a way that conclusions apply to $f$ and all its subnetworks. For instance, if $G(f)$ has no cycle, then the interaction graph $G(h)$ every subnetwork $h$ of $f$ has no cycle (since $G(h)\subseteq G(f)$), and by Robert's theorem, every subnetwork $h$ of $f$ has a unique fixed point. Such a remark is also valid for Theorem~\ref{thm:non-exp}, because if $f$ is non-expansive then all its subnetworks are non-expansive too. 
\end{remark}
%%%%%%%%%%%%%%%%%%

%%%%%%%%%%%%%%%%%%
\begin{remark}
Using the previous remark, we deduce from Theorem~\ref{thm:shihdong} (resp. Theorem \ref{thm:remy}) that if $Gf(x)$ has no cycle (resp. no positive cycle) for all $x\in\B^V$, then every subnetwork of $f$ has a unique (resp. at most one) fixed point, and thus, following Proposition~\ref{pro:critical_1}, $\Gamma(f)$ has a unique attractor (resp. at most one attractor). 
\end{remark}
%%%%%%%%%%%%%%%%%%

%%%%%%%%%%%%%%%%%%
\begin{remark}
Proceeding in a similar way, we deduce from Theorem~\ref{thm:non-exp} and Proposition~\ref{pro:critical_1} the following local version of second Thomas' rule for non-expansive networks: {\em If $f$ is non-expansive and if $Gf(x)$ has no negative cycle for all $x\in\B^V$, then $\Gamma(f)$ has no cyclic attractors.} 
\end{remark}
%%%%%%%%%%%%%%%%%%

%%%%%%%%%%%%%%%%%%%%%%%%%%%%%%%%%%%%%%%%%%%%%%%%%%%%%%%%%%%%%%%%%%%%%%
%%%%%%%%%%%%%%%%%%%%%%%%%%%%%%%%%%%%%%%%%%%%%%%%%%%%%%%%%%%%%%%%%%%%%%
\section{A forbidden subnetwork theorem}\label{sec:main}
%%%%%%%%%%%%%%%%%%%%%%%%%%%%%%%%%%%%%%%%%%%%%%%%%%%%%%%%%%%%%%%%%%%%%%
%%%%%%%%%%%%%%%%%%%%%%%%%%%%%%%%%%%%%%%%%%%%%%%%%%%%%%%%%%%%%%%%%%%%%%

In this section, we introduce the class $\F$ of even- and odd-self dual networks, and we prove that it has the following property: Every subnetworks of $f$ (and $f$ itself in particular) has a unique fixed point if and only if $f$ has no subnetwork in $\F$.  

\medskip
We say that $f$ is {\bf self-dual} if 
\[
\forall x\in\B^V,\qquad f(x\o 1)=f(x)\o 1.
\]
Equivalently, $f$ is self-dual if $\f(x\o 1)=\f(x)$ for all $x\in\B^V$. 

\medskip
We say that $f$ is {\bf even} if the image set of $\f$ is the set of even points of~$\B^V$, that is,
\[
\f(\B^V)=\{x\in\B^V\,|\,\d{x}\text{ is even}\}
\]
and similarly, we say that $f$ is {\bf odd} if 
\[
\f(\B^V)=\{x\in\B^V\,|\,\d{x}\text{ is odd}\}.
\]
Thus, if $f$ is even, then there exists $x\in \B^V$ such that $\f(x)=0$, which is equivalent to say that $f(x)=x$. Hence, even networks have at least one fixed point. Obviously, odd networks have no fixed point. 

\medskip
We say that $f$ is {\bf even-self-dual} (resp. {\bf odd-self-dual}) if it is both even (resp. odd) and self dual. We will often implicitly use the following characterization: $f$ is even-self-dual (resp. odd-self-dual) if and only if 
\[
\forall z\in\B^V\text{ s.t. }\d{z}\text{ is even (resp. odd)},~\exists x\in\B^V\text{ s.t. }\f^{-1}(z)=\{x,x\o 1\}.
\]
It follows that if $f$ is even-self-dual then it has exactly two fixed points. 

\medskip
Our interest for even- or odd-self-dual networks lies in the following theorem, which is the main result of this paper. 

%%%%%%%%%%%%%%%%%%
\begin{theorem}\label{main}
If $f$ has no even- or odd-self-dual subnetwork, then the conjugate of $f$ is a bijection.
\end{theorem}
%%%%%%%%%%%%%%%%%%

The proof needs the following two lemmas.

%%%%%%%%%%%%%%%%%%
\begin{lemma}\label{lem1}
Let $X$ be a non-empty subset of $\B^V$ and 
\[
N(X)=\{x\o e_i\,|\,x\in X,\,i\in V\}.
\]
If $X$ and $N(X)$ are disjoint and $|X|\geq |N(X)|$, then $X$ is either the set of even points of $\B^V$ or the set of odd points of~$\B^V$.
\end{lemma}
%%%%%%%%%%%%%%%%%%

%%%%%%%%%%%%%%%%%%
\begin{proof}
by induction on $|V|$. The case $|V|=1$ is obvious. So suppose that $|V|>1$. Let $X$ be a non-empty subset of $\B^V$ satisfying the conditions of the statement. Let $i\in V$ and $\a\in\B$. For all $Y\subseteq \B^V$, let us denote by $Y^\a$ be the subset of $\B^{V\setminus \{i\}}$ defined by $Y^\a=\{x_{-i}\,|\,x\in Y,x_i=\a\}$. 

\medskip
We first prove that $N(X^\a)\subseteq N(X)^\a$ and $X^\a\cap N(X^\a)=\emptyset$. Let $x\in\B^V$ with $x_i=\a$ be such that $x_{-i}\in N(X^\a)$. To prove that $N(X^\a)\subseteq N(X)^\a$, it is sufficient to prove that $x_{-i}\in N(X)^\a$. Since $x_{-i}\in N(X^\a)$, there exists $y\in\B^V$ with $y_i=\a$ and $j\in V$ with $j\neq i$ such that $y_{-i}\in X^\a$ and $x_{-i}=y_{-i}\o e_j$. So $x=y\o e_j$, and since $y_i=\a$, we have $y\in X$. Hence $x\in N(X)$ and since $x_i=\a$, we have $x_{-i}\in N(X)^\a$. We now prove that $X^\a\cap N(X^\a)=\emptyset$. Indeed, otherwise, there exists $x\in\B^V$ with $x_i=\a$ such that $x_{-i}\in X^\a\cap N(X^\a)$. Since $N(X^\a)\subseteq N(X)^\a$, we have $x_{-i}\in X^\a\cap N(X)^\a$, and since $x_i=\a$, we deduce that $x\in X\cap N(X)$, a contradiction.

\medskip
Since $N(X^\a)\subseteq N(X)^\a$, we have 
\[
|X|=|X^{0}|+|X^{1}|\geq |N(X)|=|N(X)^0|+|N(X)^1|\geq |N(X^{0})|+|N(X^{1})|.
\] 
So $|X^{0}|\geq |N(X^{0})|$ or $|X^{1}|\geq |N(X^{1})|$. Suppose that $|X^{0}|\geq |N(X^{0})|$, the other case being similar. Since $X^0\cap N(X^0)=\emptyset$, by induction hypothesis $X^{0}$ is either the set of even points of $\B^{V\setminus \{i\}}$ or the set of odd points of $\B^{V\setminus \{i\}}$. So in both cases, we have $|X^{0}|=|N(X^{0})|=2^{|V|-1}$. We deduce that $|X^{1}|\geq |N(X^{1})|$, and so, by induction hypothesis, $X^{1}$ is either the set of even points of $\B^{V\setminus \{i\}}$ or the set of odd points of $\B^{V\setminus \{i\}}$. But $X^{0}$ and $X^{1}$ are disjointed: For all $x\in\B^V$, if $x_{-i}\in X^{0}\cap X^{1}$, then $x^{i0}$ and $x^{i1}$ are two points of $X$, and $x^{i1}=x^{i0}\o e_i\in N(X)$, a contradiction. So if $X^{0}$ is the set of even (resp. odd) points of $\B^{V\setminus \{i\}}$, then $X^{1}$ is the set of odd (resp. even) points of $\B^{V\setminus \{i\}}$, and we deduce that $X$ is the set of even (resp. odd) points of $\B^V$.
\end{proof}
%%%%%%%%%%%%%%%%%%

%%%%%%%%%%%%%%%%%%
\begin{lemma}\label{lem2}
Suppose that the conjugate of every immediate subnetwork of a network $f$ is a bijection. If the conjugate of $f$ is not a bijection, then $f$ is even- or odd-self-dual.
\end{lemma}
%%%%%%%%%%%%%%%%%%

%%%%%%%%%%%%%%%%%%
\begin{proof}
Suppose that $f:\B^V\to\B^V$ satisfies the conditions of the statement, and suppose that the conjugate $\f$ of $f$ is not a bijection. Let
\[
X=\f(\B^V),\qquad \bar X=\B^V\setminus X. 
\]
Since $\f$ is not a bijection, $\bar X$ is not empty. 

\medskip
Let us first prove that 
\[\tag{$*$}
\forall x\in \bar X,~\forall i\in V,\qquad |\f^{-1}(x\o e_i)|=2.
\]
Let $x\in \bar X$ and $i\in V$. By hypothesis, $\f^{i0}$ is a bijection, so there exists a unique point $y\in\B^V$ with $y_i=0$ such that $\f^{i0}(y_{-i})=x_{-i}$. Then, $\f(y)_{-i}=\f(y^{i0})_{-i}=\f^{i0}(y_{-i})=x_{-i}$. In other words $\f(y)\in\{x,x\o e_i\}$. Since $x\in \bar X$ we have $\f(y)\neq x$ and it follows that $\f(y)= x\o e_i$. Hence, we have proved that there exists a unique point $y\in\B^V$ such that $y_i=0$ and $\f(y)=x\o e_i$, and we prove with similar arguments that there exists a unique point $z\in\B^V$ such that $z_i=1$ and $\f(z)=x\o e_i$. This proves ($*$).

\medskip
We are now in position to prove that $f$ is even or odd. Let
\[
N(\bar X)=\{x\o e_i\,|\,x\in \bar X,\,i\in V\}.
\]
Following ($*$) we have $N(\bar X)\subseteq X$, and we deduce that 
\begin{multline*}
|\f^{-1}(X)|=|\f^{-1}(N(\bar X))|+|\f^{-1}(X\setminus N(\bar X))|\\[1mm]
\geq|\f^{-1}(N(\bar X))|+|X\setminus N(\bar X)|\\[1mm]=|\f^{-1}(N(\bar X))|+|X|-|N(\bar X)|.
\end{multline*}
Again following ($*$), $|\f^{-1}(N(\bar X))|=2|N(\bar X)|$ and we deduce that 
\begin{multline*}
|X|+|\bar X|=2^{|V|}=|\f^{-1}(X)|\geq 2|N(\bar X)|+|X|-|N(\bar X)|
=|N(\bar X)|+|X|.
\end{multline*}
Therefore, $|\bar X|\geq |N(\bar X)|$, and since $N(\bar X)\subseteq X=\B^{|V|}\setminus \bar X$, we have $\bar X\cap N(\bar X)=\emptyset$. So according to Lemma~1, $\bar X$ is either the set of even points of $\B^{|V|}$ or the set of odd points of $\B^{|V|}$. We deduce that in the first (second) case, $X$ is the set of odd (even) points of $\B^{|V|}$. Thus, $f$ is even or odd.

\medskip
It remains to prove that $f$ is self-dual. Let $x\in \B^V$. For all $i\in V$, since $\d{\f(x)}$ and $\d{\f(x)\o e_i}$ have not the same parity, and since $f$ is even or odd, we have $\f(x)\o e_i\in \bar X$. Thus, according to ($*$), the preimage of $(\f(x)\o e_i)\o e_i=\f(x)$ by $\f$ is of cardinality two. Consequently, there exists a point $y\in\B^{|V|}$, distinct from $x$, such that $\f(y)=\f(x)$. Let us proved that $x=y\o 1$. Indeed, if $x_i=y_i=0$ for some $i\in V$, then $\f^{i0}(x_{-i})=\f(x)_{-i}=\f(y)_{-i}=\f^{i0}(y_{-i})$. Since $x\neq y$, we deduce that $\f^{i0}$ is not a bijection, a contradiction. We show similarly that if $x_i=y_i=1$, then $\f^{i1}$ is not a bijection. So $x=y\o 1$. Consequently, $\f(x\o 1)=\f(x)$, and we deduce that $f$ is self-dual.
\end{proof}
%%%%%%%%%%%%%%%%%%

%%%%%%%%%%%%%%%%%%
\begin{proof}[Proof of Theorem~\ref{main}]
by induction on $|V|$. The case $|V|=1$ is obvious. So suppose that $|V|>1$ and suppose  that $f$ has no even- or odd-self-dual subnetwork. Under this condition, $f$ is neither even-self-dual nor odd-self-dual (since $f$ is a subnetwork of $f$), and every immediate subnetwork of $f$ has no even- or odd-self-dual subnetwork. So, by induction hypothesis, the dual of every strict subnetwork of $f$ is a bijection, and we deduce from Lemma~2 that the dual of $f$ is a bijection. 
\end{proof}
%%%%%%%%%%%%%%%%%%

%%%%%%%%%%%%%%%%%%
\begin{corollary}\label{cor1}
The conjugate of each subnetwork of $f$ is a bijection if and only if $f$ has no even- or odd-self-dual subnetworks.
\end{corollary}
%%%%%%%%%%%%%%%%%%

%%%%%%%%%%%%%%%%%%
\begin{proof}
If $f$ has no even- or odd-self-dual subnetwork, then every subnetwork $h$ of $f$ has no even- or odd-self-dual subnetwork, and according to Theorem~\ref{main}, the conjugate of $h$ is a bijection.  Conversely, if the conjugate of each subnetwork of $f$ is a bijection, then $f$ has clearly no even- or odd-self-dual subnetwork (since if a network is even or odd, its conjugate sends $\B^V$ to a subset of $\B^V$ of cardinality $|\B^V|/2$).
\end{proof}
%%%%%%%%%%%%%%%%%%

If $\f$ is a bijection then there is a unique point $x\in\B^V$ such that $\f(x)=0$, and this point is thus the unique fixed point of $f$. As an immediate consequence of this property and the previous corollary, we obtain the characterization mentioned at the beginning of the section. 

%%%%%%%%%%%%%%%%%%
\begin{corollary}\label{cor2}
Each subnetwork of $f$ has a unique fixed point ($f$ in particular) if and only if $f$ has no even- or odd-self-dual subnetworks.
\end{corollary}
%%%%%%%%%%%%%%%%%%

%%%%%%%%%%%%%%%%%%
\begin{remark}
As an immediate consequence of the two previous corollary, we get the following property, independently proved by Ruet in \cite{Ruet11}: {\em Each subnetwork of $f$ has a unique fixed point if and only if the conjugate of each subnetwork of $f$ is a bijection.} 
\end{remark}
%%%%%%%%%%%%%%%%%%

%%%%%%%%%%%%%%%%%%
\begin{example}\label{ex1} 	 
Consider the following network $f$ on $\{1,2,3\}$~\footnote{In all the examples, network components are integers, and if $V$ is a set of $n$ integers $i_1\!<\!i_2\!<\!\cdots\!<\!i_n$, then for all $x\in\B^V$ we write $x=(x_{i_1},x_{i_2},\dots,x_{i_n})$ or $x=x_{i_1}x_{i_2}\dots x_{i_n}$.}:
\[
\begin{array}{l}
f:\B^{\{1,2,3\}}\to\B^{\{1,2,3\}}
\end{array}
\qquad
\begin{array}{l}
f_1(x)=\overline{x_2}\land x_3\\[1mm]
f_2(x)=\overline{x_3}\land x_1\\[1mm]
f_3(x)=\overline{x_1}\land x_2.
\end{array}
\]
The table of $f$ and $\f$ are:
\[
\begin{array}{ccc}
x & f(x) & \f(x) \\\hline
000& 000 & 000\\ 
001& 100 & 101\\
010& 001 & 011\\
011& 001 & 010\\
100& 010 & 110\\
101& 100 & 001\\
110& 010 & 100\\
111& 000 & 111
\end{array}
\]
The six immediate subnetworks of $f$ are: 
\[
\begin{array}{l}
f^{10}:\B^{\{2,3\}}\to\B^{\{2,3\}}
\end{array}
\qquad
\begin{array}{l}
f^{10}_2(x)=\overline{x_3}\land 0=0\\[1mm]
f^{10}_3(x)=\overline{0}\land x_2=x_2
\end{array}
\]
\[
\begin{array}{l}
f^{11}:\B^{\{2,3\}}\to\B^{\{2,3\}}
\end{array}
\qquad
\begin{array}{l}
f^{11}_2(x)=\overline{x_3}\land 1=\overline{x_3}\\[1mm]
f^{11}_3(x)=\overline{1}\land x_2=0
\end{array}
\]
\[
\begin{array}{l}
f^{20}:\B^{\{1,3\}}\to\B^{\{1,3\}}
\end{array}
\qquad
\begin{array}{l}
f^{20}_1(x)=\overline{0}\land x_3=x_3\\[1mm]
f^{20}_3(x)=\overline{x_1}\land 0=0
\end{array}
\]
\[
\begin{array}{l}
f^{21}:\B^{\{1,3\}}\to\B^{\{1,3\}}
\end{array}
\qquad
\begin{array}{l}
f^{21}_1(x)=\overline{1}\land x_3=0\\[1mm]
f^{21}_3(x)=\overline{x_1}\land 1=\overline{x_1}
\end{array}
\]
\[
\begin{array}{l}
f^{30}:\B^{\{1,2\}}\to\B^{\{1,2\}}
\end{array}
\qquad
\begin{array}{l}
f^{30}_1(x)=\overline{x_2}\land 0=0\\[1mm]
f^{30}_2(x)=\overline{0}\land x_1=x_1
\end{array}
\]
\[
\begin{array}{l}
f^{31}:\B^{\{1,2\}}\to\B^{\{1,2\}}
\end{array}
\qquad
\begin{array}{l}
f^{31}_1(x)=\overline{x_2}\land 1=\overline{x_2}\\[1mm]
f^{31}_2(x)=\overline{1}\land x_1=0
\end{array}
\]
So each immediate subnetwork $f^{i\a}$ of $f$ has one component fixed to zero, so $f$ has no self-dual immediate subnetwork. Furthermore, each immediate subnetwork of $f^{i\a}$ is a constant ($0$), and thus is not self-dual. Furthermore, $f$ is not self-dual since $f(000)=f(111)=111$. Hence, $f$ has no self-dual subnetwork, and we deduce from Theorem~\ref{main} that the conjugate of $\tilde f$ of $f$ is a bijection. This can be easily verified on the table given above. 
\end{example}
%%%%%%%%%%%%%%%%%%

%%%%%%%%%%%%%%%%%%%%%%%%%%%%%%%%%%%%%%%%%%%%%%%%%%%%%%%%%%%%%%%%%%%%%%
%%%%%%%%%%%%%%%%%%%%%%%%%%%%%%%%%%%%%%%%%%%%%%%%%%%%%%%%%%%%%%%%%%%%%%
\section{Generalization of Shih-Dong's theorem}\label{sec:shih}
%%%%%%%%%%%%%%%%%%%%%%%%%%%%%%%%%%%%%%%%%%%%%%%%%%%%%%%%%%%%%%%%%%%%%%
%%%%%%%%%%%%%%%%%%%%%%%%%%%%%%%%%%%%%%%%%%%%%%%%%%%%%%%%%%%%%%%%%%%%%%

In this section, we show, using Theorem~\ref{main}, that the condition ``$Gf(x)$ has no cycles for all $x$'' in Shih-Dong's theorem (Theorem~\ref{thm:shihdong}) can be weakened into a condition of the form ``$Gf(x)$ has short cycles for few points $x$''. The exact statement is given after the following useful proposition. 

%%%%%%%%%%%%%%%%%%
\begin{proposition}\label{odd}
If $f$ is even or odd, then for every $x\in\B^V$ the out-degree of each vertex of $Gf(x)$ is odd. In particular, $Gf(x)$ has a cycle.
\end{proposition}
%%%%%%%%%%%%%%%%%%

%%%%%%%%%%%%%%%%%%
\begin{proof}
Let $j\in V$ and let $d$ be the out-degree of $j$ in $Gf(x)$. Since $d$ equals the number of
$i\in V$ such that $|f_{ij}(x)|=1$, and since 
\[
|f_{ij}(x)|=f_i(x^{j1})\o f_i(x^{j0})=f_i(x)\o f_i(x\o e_j),
\]
we have 
\[
\begin{array}{rcl}
d&=&\d{f(x)\o f(x\o e_j)}\\[1mm]
&=&\d{(x\o \f(x))\o ((x\o e_j)\o \f(x\o e_j))}\\[1mm]
&=&\d{\f(x)\o \f(x\o e_j)\o e_j}.
\end{array}
\]
So the parity of $d$ is the parity of $\d{\f(x)}+\d{\f(x\o e_i)}+1$. Hence, if $f$ is
even or odd, then $\d{\f(x)}$ and $\d{\f(x\o e_i)}$ have the same parity, so $\d{\f(x)}+\d{\f(x\o e_i)}$ is even, and it follows that $d$ is odd. 
\end{proof}
%%%%%%%%%%%%%%%%%%

%%%%%%%%%%%%%%%%%%
\begin{corollary}\label{cor:shih}
If, for every $1\leq k\leq |V|$, there exists at most $2^k-1$ points $x\in
\B^V$ such that $Gf(x)$ has a cycle of length at most $k$, then $f$ has a
unique fixed point.
\end{corollary}
%%%%%%%%%%%%%%%%%%

%%%%%%%%%%%%%%%%%%
\begin{proof}
According to Theorem \ref{main}, it is sufficient to prove, by induction on
$|V|$, that if $f$ satisfies the conditions of the statement, then $f$
has no even- or odd-self-dual subnetwork. The case $|V|=1$ is
obvious, so suppose that $|V|>1$. Suppose also that $f$ satisfies the conditions of
the statement. Let $i\in V$ and $\a\in\B$. Since $Gf^{i\a}(x_{-i})$ is the subgraph of
$Gf(x^{i\a})$ for all $x\in\B^V$ (cf. Proposition~\ref{pro:subgraph}), $f^{i\a}$  satisfies the condition of the theorem. Thus, by induction hypothesis, $f^{i\a}$ has no even- or
odd-self-dual subnetwork. So $f$ has no even- or odd-self-dual
strict subnetwork. If $f$ is itself even- or odd-self-dual, then by
Proposition~\ref{odd}, $Gf(x)$ has a cycle for every
$x\in\B^{V}$, so $f$ does not satisfy that conditions of the
statement (for $k=|V|$). Therefore, $f$ has no even- or odd-self-dual
subnetwork. 
\end{proof}
%%%%%%%%%%%%%%%%%%

%%%%%%%%%%%%%%%%%%
\begin{example}[Continuation of Example~\ref{ex1}]\label{ex2}
Take again the 3-dimensional network $f$ defined by
\[
\begin{array}{l}
f_1(x)=\overline{x_2}\land x_3\\[1mm]
f_2(x)=\overline{x_3}\land x_1\\[1mm]
f_3(x)=\overline{x_1}\land x_2.
\end{array}
\]
We have seen that $f$ has no self-dual subnetwork. So it satisfies the
conditions of Theorem~\ref{main}, but not the conditions of
Shih-Dong's theorem. Indeed, $Gf(000)$ and $Gf(111)$ have a cycle \footnote{Arrows correspond to positive arcs and bars to negative arcs.}:
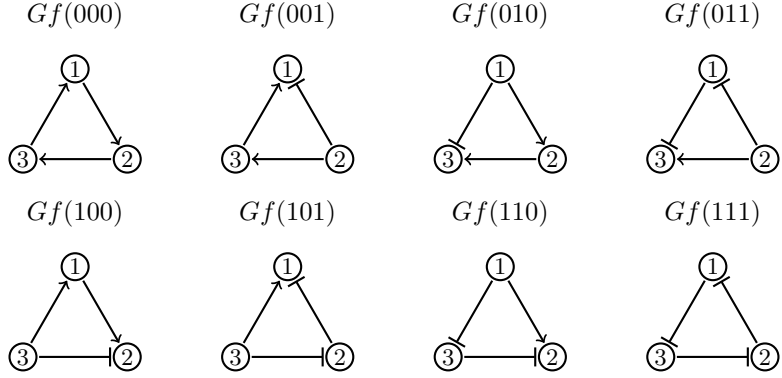
\[
\begin{array}{cccc}
\begin{array}{c}%000
Gf(000)\\[3mm]
\begin{array}{c}
\begin{tikzpicture}
\pgfmathparse{0.8}
\node[outer sep=1,inner sep=1,circle,draw,thick] (1) at (90:\pgfmathresult){\small$1$};
\node[outer sep=1,inner sep=1,circle,draw,thick] (2) at (-30:\pgfmathresult){\small$2$};
\node[outer sep=1,inner sep=1,circle,draw,thick] (3) at (210:\pgfmathresult){\small$3$};
\path[thick]
(1) edge[->] (2)
(2) edge[->] (3)
(3) edge[->] (1)
;
\end{tikzpicture}
\end{array}
\end{array}
&
\begin{array}{c}%001
Gf(001)\\[3mm]
\begin{array}{c}
\begin{tikzpicture}
\pgfmathparse{0.8}
\node[outer sep=1,inner sep=1,circle,draw,thick] (1) at (90:\pgfmathresult){\small$1$};
\node[outer sep=1,inner sep=1,circle,draw,thick] (2) at (-30:\pgfmathresult){\small$2$};
\node[outer sep=1,inner sep=1,circle,draw,thick] (3) at (210:\pgfmathresult){\small$3$};
\path[thick]
(2) edge[-|] (1)
(2) edge[->] (3)
(3) edge[->] (1)
;
\end{tikzpicture}
\end{array}
\end{array}
&
\begin{array}{c}%010
Gf(010)\\[3mm]
\begin{array}{c}
\begin{tikzpicture}
\pgfmathparse{0.8}
\node[outer sep=1,inner sep=1,circle,draw,thick] (1) at (90:\pgfmathresult){\small$1$};
\node[outer sep=1,inner sep=1,circle,draw,thick] (2) at (-30:\pgfmathresult){\small$2$};
\node[outer sep=1,inner sep=1,circle,draw,thick] (3) at (210:\pgfmathresult){\small$3$};
\path[thick]
(1) edge[->] (2)
(1) edge[-|] (3)
(2) edge[->] (3)
;
\end{tikzpicture}
\end{array}
\end{array}
&
\begin{array}{c}%011
Gf(011)\\[3mm]
\begin{array}{c}
\begin{tikzpicture}
\pgfmathparse{0.8}
\node[outer sep=1,inner sep=1,circle,draw,thick] (1) at (90:\pgfmathresult){\small$1$};
\node[outer sep=1,inner sep=1,circle,draw,thick] (2) at (-30:\pgfmathresult){\small$2$};
\node[outer sep=1,inner sep=1,circle,draw,thick] (3) at (210:\pgfmathresult){\small$3$};
\path[thick]
(1) edge[-|] (3)
(2) edge[-|] (1)
(2) edge[->] (3)
;
\end{tikzpicture}
\end{array}
\end{array}
\\[12mm]
\begin{array}{c}%100
Gf(100)\\[3mm]
\begin{array}{c}
\begin{tikzpicture}
\pgfmathparse{0.8}
\node[outer sep=1,inner sep=1,circle,draw,thick] (1) at (90:\pgfmathresult){\small$1$};
\node[outer sep=1,inner sep=1,circle,draw,thick] (2) at (-30:\pgfmathresult){\small$2$};
\node[outer sep=1,inner sep=1,circle,draw,thick] (3) at (210:\pgfmathresult){\small$3$};
\path[thick]
(1) edge[->] (2)
(3) edge[-|] (2)
(3) edge[->] (1)
;
\end{tikzpicture}
\end{array}
\end{array}
&
\begin{array}{c}%101
Gf(101)\\[3mm]
\begin{array}{c}
\begin{tikzpicture}
\pgfmathparse{0.8}
\node[outer sep=1,inner sep=1,circle,draw,thick] (1) at (90:\pgfmathresult){\small$1$};
\node[outer sep=1,inner sep=1,circle,draw,thick] (2) at (-30:\pgfmathresult){\small$2$};
\node[outer sep=1,inner sep=1,circle,draw,thick] (3) at (210:\pgfmathresult){\small$3$};
\path[thick]
(2) edge[-|] (1)
(3) edge[-|] (2)
(3) edge[->] (1)
;
\end{tikzpicture}
\end{array}
\end{array}
&
\begin{array}{c}%110
Gf(110)\\[3mm]
\begin{array}{c}
\begin{tikzpicture}
\pgfmathparse{0.8}
\node[outer sep=1,inner sep=1,circle,draw,thick] (1) at (90:\pgfmathresult){\small$1$};
\node[outer sep=1,inner sep=1,circle,draw,thick] (2) at (-30:\pgfmathresult){\small$2$};
\node[outer sep=1,inner sep=1,circle,draw,thick] (3) at (210:\pgfmathresult){\small$3$};
\path[thick]
(1) edge[->] (2)
(1) edge[-|] (3)
(3) edge[-|] (2)
;
\end{tikzpicture}
\end{array}
\end{array}
&
\begin{array}{c}%111
Gf(111)\\[3mm]
\begin{array}{c}
\begin{tikzpicture}
\pgfmathparse{0.8}
\node[outer sep=1,inner sep=1,circle,draw,thick] (1) at (90:\pgfmathresult){\small$1$};
\node[outer sep=1,inner sep=1,circle,draw,thick] (2) at (-30:\pgfmathresult){\small$2$};
\node[outer sep=1,inner sep=1,circle,draw,thick] (3) at (210:\pgfmathresult){\small$3$};
\path[thick]
(1) edge[-|] (3)
(3) edge[-|] (2)
(2) edge[-|] (1)
;
\end{tikzpicture}
\end{array}
\end{array}
\end{array}
\]
However, $f$ satisfies the condition of Corollary~\ref{cor:shih} (there is $0<2^1$ point $x$ such that $Gf(x)$ has a cycle of length~at most $1$, $0<2^2$ point $x$ such that $Gf(x)$ has a cycle of length at most $2$, and $2<2^3$ points $x$ such that $Gf(x)$ has a cycle of length at most $3$). From the local interactions graphs given above, we deduce that the global interaction graph $G(f)$ of the network is the following:
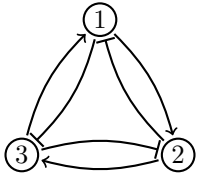
\[
\begin{tikzpicture} 
\pgfmathparse{1.9}
\node[outer sep=1.5,inner sep=1.5,circle,draw,thick] (1) at (90:1.2){$1$};
\node[outer sep=1.5,inner sep=1.5,circle,draw,thick] (2) at (-30:1.2){$2$};
\node[outer sep=1.5,inner sep=1.5,circle,draw,thick] (3) at (210:1.2){$3$};
\path[thick]
(1) edge[bend left=15,->] (2)
(2) edge[bend left=15,->] (3)
(3) edge[bend left=15,->] (1)
(1) edge[bend left=15,-|] (3)
(3) edge[bend left=15,-|] (2)
(2) edge[bend left=15,-|] (1)
;
\end{tikzpicture}
\]
\end{example}
%%%%%%%%%%%%%%%%%%

In addition to Proposition~\ref{odd}, we have the following property on the  structure of the interactions of even- and odd-self-dual networks.  

%%%%%%%%%%%%%%%%%%
\begin{proposition}
If $f$ is a critical even- or odd-self-dual network then $G(f)$ is strongly connected.  
\end{proposition}
%%%%%%%%%%%%%%%%%%

%%%%%%%%%%%%%%%%%%
\begin{proof}
Suppose that $f$ is critical even- or odd-self-dual. If $G(f)$ is not strongly connected, then it has an initial strongly connected component $I$ (no arc from $V\setminus I$ to $I$) strictly included in $V$. Let $h$ be the subnetwork of $f$ induced by some point $z\in\B^I$. Since $f$ is critical and since $h$ is a strict subnetwork, according to Theorem~\ref{main}, $\h$ is a bijection. Thus there exists $x,y\in\B^V$ with $x|_I=y|_I=z$, such that $\h(x_{-I})$ and $\h(y_{-I})$ have not the same parity. Since $\f(x)_{-I}=\h(x_{-I})$ and $\f(y)_{-I}=\h(y_{-I})$ we have    
\[
\d{\f(x)}=\d{\f(x)|_I}+\d{\h(x_{-I})},\qquad 
\d{\f(y)}=\d{\f(y)|_I}+\d{\h(y_{-I})}.
\] 
Since $G(f)$ has no arc from $V\setminus I$ to $I$, and since $x|_I=y|_I$ we have $f(x)|_I=f(y)|_I$ and thus $\f(x)|_I=\f(y)|_I$. Thus $\f(x)$ and $\f(y)$ have not the same parity, a contradiction.
\end{proof}
%%%%%%%%%%%%%%%%%%

%%%%%%%%%%%%%%%%%%%%%%%%%%%%%%%%%%%%%%%%%%%%%%%%%%%%%%%%%%%%%%%%%%%%%%
%%%%%%%%%%%%%%%%%%%%%%%%%%%%%%%%%%%%%%%%%%%%%%%%%%%%%%%%%%%%%%%%%%%%%%
\section{Weak asynchronous convergence}\label{sec:asyn}
%%%%%%%%%%%%%%%%%%%%%%%%%%%%%%%%%%%%%%%%%%%%%%%%%%%%%%%%%%%%%%%%%%%%%%
%%%%%%%%%%%%%%%%%%%%%%%%%%%%%%%%%%%%%%%%%%%%%%%%%%%%%%%%%%%%%%%%%%%%%%

We say that the asynchronous state graph $\Gamma(f)$ describes a {\bf strong asynchronous convergence} toward a unique fixed $x$ if $\Gamma(f)$ is acyclic and admits $x$ as unique attractor. We say that $\Gamma(f)$ describes a {\bf weak asynchronous convergence} toward a unique fixed point $x$ if $\Gamma(f)$ admits $x$ as unique attractor (equivalently, $f$ has a unique fixed point $x$ and $\Gamma(f)$ has a path from any point $y$ to $x$). The following corollary shows that the absence of even- or odd-self-dual subnetwork implies a weak asynchronous convergence toward a unique fixed point. 

%%%%%%%%%%%%%%%%%%
\begin{corollary}\label{corThomas} 
If $f$ has no even or odd self-dual subnetwork, then $f$ has a unique fixed point $x$, and for all $y\in\B^V$, the asynchronous state graph of $f$ contains a path from $y$ to $x$ of length $d(x,y)$.
\end{corollary}
%%%%%%%%%%%%%%%%%%

%%%%%%%%%%%%%%%%%%
\begin{remark}
By definition, if $x\to y$ is an arc of the asynchronous state graph, then $d(x,y)=1$. Hence, path from a point $x$ to a point $y$ cannot be of length strictly less than $d(x,y)$; a path from $x$ to $y$ of length $d(x,y)$ can thus be seen as a {\em shortest} or {\em straight} path.
\end{remark}
%%%%%%%%%%%%%%%%%%

%%%%%%%%%%%%%%%%%%
\begin{proof}[Proof of Corollary~\ref{corThomas}]
By induction on $|V|$. The case $|V|=1$ is obvious, so suppose that $|V|>1$ and that $f$ has no even or odd self-dual subnetwork. By Theorem~\ref{main}, $f$ has a unique fixed point $x$. Let $y\in\B^V$. Suppose first that there exists $i\in V$ such that $x_i=y_i=0$. Then $x_{-i}$ is the unique fixed point of $f^{i0}$. So, by induction hypothesis, $\Gamma(f^{i0})$ has a path from $y_{-i}$ to $x_{-i}$ of length $d(x_{-i},y_{-i})$. Since $x_i=y_i=0$, we deduce from Proposition~\ref{pro:subdynamics} that $\Gamma(f)$ has a path from $y$ to $x$ of length $d(x_{-i},y_{-i})=d(x,y)$. The case $x_i=y_i=1$ is similar. So, finally, suppose that $y=x\o 1$. Since $y$ is not a fixed point, there exists $i\in V$ such that $f_i(y)\neq y_i$. Then, $\Gamma(f)$ has an arc from $y$ to $z=y\o e_i$. So $z_i=x_i$, and as previously, we deduce that $\Gamma(f)$ has a path from $z$ to $x$ of length $d(x,z)$. This path together with the arc $y\to z$ forms a path from $y$ to $x$ of length $d(x,z)+1=d(x,y)$.
\end{proof}
%%%%%%%%%%%%%%%%%%

%%%%%%%%%%%%%%%%%%
\begin{remark}
According to Proposition~\ref{pro:subdynamics}, the asynchronous state graph $\Gamma(h)$ of each subnetwork $h$ of $f$ is a subgraph of $\Gamma(f)$ induced by some subcube of $\B^V$. Hence, one can see $\Gamma(h)$ as a ``{\em dynamical module}'' of $\Gamma(f)$. An interpretation of the previous corollary is then that {\em the asynchronous state graphs of even- and odd-self-dual networks are ``dynamical modules'' that are necessary for the ``emergence'' of ``complex'' asynchronous behaviors}, because in their absence the dynamics is ``simple'': weak asynchronous convergence toward a unique fixed point. 
\end{remark}
%%%%%%%%%%%%%%%%%%

%%%%%%%%%%%%%%%%%%
\begin{example}[Continuation of Example~\ref{ex1}]\label{ex3}
Take again the $3$-dimensional network $f$ defined in
Example~\ref{ex1}, which has no self-dual subnetwork. 
\[
\begin{array}{l}
f_1(x)=\overline{x_2}\land x_3\\[1mm]
f_2(x)=\overline{x_3}\land x_1\\[1mm]
f_3(x)=\overline{x_1}\land x_2.
\end{array}
\]
The asynchronous
state graph $\Gamma(f)$ of $f$ is the following:
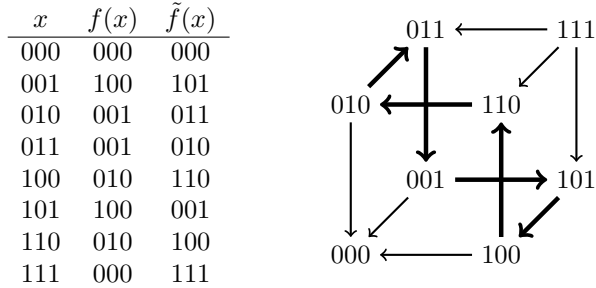
\[
\begin{array}{ccc}
x & f(x) & \f(x) \\\hline
000& 000 & 000\\ 
001& 100 & 101\\
010& 001 & 011\\
011& 001 & 010\\
100& 010 & 110\\
101& 100 & 001\\
110& 010 & 100\\
111& 000 & 111
\end{array}
\qquad\quad
\begin{array}{c}
\begin{tikzpicture}
\node (000) at (-1.5,-1.5){$000$};
\node (001) at (-0.5,-0.5){$001$};
\node (010) at (-1.5,0.5){$010$};
\node (011) at (-0.5,1.5){$011$};
\node (100) at (0.5,-1.5){$100$};
\node (101) at (1.5,-0.5){$101$};
\node (110) at (0.5,0.5){$110$};
\node (111) at (1.5,1.5){$111$};
\path[thick,->]
(001) edge[line width=1.7] (101)
(001) edge (000)
(010) edge (000)
(010) edge[line width=1.7] (011)
(011) edge[line width=1.7] (001)
(100) edge (000)
(100) edge[line width=1.7] (110)
(101) edge[line width=1.7] (100)
(110) edge[line width=1.7] (010)
(111) edge (011)
(111) edge (101)
(111) edge (110)
;
\end{tikzpicture}
\end{array}
\]
In agreement with Corollary~\ref{corThomas}, there exists, from any initial point, a shortest path leading to the unique fixed point of $f$ (the point $000$): the asynchronous state graph describes a weak asynchronous convergence (by shortest paths) toward a unique fixed point. However, $\Gamma(f)$ has a cycle (of length 6), so every path does not lead to the unique fixed point: the condition ``has no even or odd self-dual subnetworks'' does no ensure a strong asynchronous convergence toward a unique fixed point.
\end{example}
%%%%%%%%%%%%%%%%%%

%%%%%%%%%%%%%%%%%%%%%%%%%%%%%%%%%%%%%%%%%%%%%%%%%%%%%%%%%%%%%%%%%%%%%%
%%%%%%%%%%%%%%%%%%%%%%%%%%%%%%%%%%%%%%%%%%%%%%%%%%%%%%%%%%%%%%%%%%%%%%
\section{Characterization by forbidden subnetworks}\label{sec:sub}
%%%%%%%%%%%%%%%%%%%%%%%%%%%%%%%%%%%%%%%%%%%%%%%%%%%%%%%%%%%%%%%%%%%%%%
%%%%%%%%%%%%%%%%%%%%%%%%%%%%%%%%%%%%%%%%%%%%%%%%%%%%%%%%%%%%%%%%%%%%%%

In this section, we are interested in characterizing networks properties by forbidden subnetworks, such as the characterization given by Corollary~\ref{cor2}. We see a network property $\P$ as a set of networks, and given a set of networks $\F$, we say that $\F$ is a {\bf set of forbidden subnetworks for} $\P$ if 
\[
f\in\P\iff\Sub(f)\cap \F=\emptyset,
\]
where $\Sub(f)$ denotes the set of subnetworks of $f$. Thus, if $\F$ is a set of forbidden subnetworks for $\P$ then $\F\cap\P=\emptyset$ and $\P$ is {\bf closed} for the subnetwork relation {\em i.e.} if $f\in\P$ then $\Sub(f)\subseteq \P$. The negation (or complement) of $\P$ is denoted $\neg\P$.

%%%%%%%%%%%%%%%%%%
\begin{proposition}\label{pro:critical_2}	
Let $\P$ be a set of networks closed for the subnetwork relation. There exists a unique smallest set $\F$ of forbidden subnetworks for $\P$. This set $\F$ is the set of networks critical for $\neg\P$. 
\end{proposition}
%%%%%%%%%%%%%%%%%%

%%%%%%%%%%%%%%%%%%
\begin{proof}
If $f\not\in\P$, then $f$ necessarily contains a subnetwork $h\not\in\P$ such that $\Sub(h)\setminus h\subseteq \P$ {\em i.e.} a subnetwork critical for $\neg\P$. Conversely, if $f\in\P$ then $\Sub(f)\subseteq \P$ and since networks critical for $\neg\P$ are in $\neg\P$, $f$ has no subnetworks critical for $\neg\P$. This proves that the set of networks critical for $\neg\P$ is a set of forbidden subnetworks for~$\P$. 

\medskip
Now, suppose that $\F$ is a set of forbidden subnetworks for $\P$, let $f$ be any network critical for $\neg\P$, and let us prove that $f\in\F$. Since every strict subnetwork of $f$ is in $\P$, $f$ has no strict subnetworks in $\F$. So if $f$ is not in $\F$ then $\Sub(f)\cap\F=\emptyset$ and we deduce that $f\in\P$, a contradiction. Thus $f\in\F$, so $\F$ contains all the networks critical for $\neg\P$. 
\end{proof}
%%%%%%%%%%%%%%%%%%

Let $\P_{=1}$ be the set of networks $f$ such that each subnetwork of $f$ has a unique fixed point, and let $\F_{=1}$ be the smallest set of forbidden subnetworks for $\P_{=1}$. Let $\esd$ and $\osd$ be the set of {\em critical} even- and odd-self-dual networks, respectively, and let $\eosd=\esd\cup\osd$. 

%%%%%%%%%%%%%%%%%%
\begin{remark}
A lot of even- or odd-self-dual networks are not critical. For instance, the network $f$ on $\{1,2,3\}$ defined by $f_1(x)=x_1\o x_2\o x_3$ and $f_2(x)=f_3(x)=x_1$ is even-self-dual, but it contains two even-self-dual strict subnetworks and two odd-self-dual strict subnetworks.
\end{remark}
%%%%%%%%%%%%%%%%%%

%%%%%%%%%%%%%%%%%%
\begin{corollary}\label{cor:P=1}
$\F_{=1}=\eosd$.
\end{corollary}
%%%%%%%%%%%%%%%%%%

%%%%%%%%%%%%%%%%%%
\begin{proof}
If $f\in\eosd$, then not strict subnetwork of $f$ is in $\eosd$ and according to  Theorem~\ref{main}, each strict subnetwork of $f$ is in $\P_{=1}$. Since $f\not\in\P_{=1}$ (because $f$ has zero or two fixed points), $f$ is  critical for $\neg\P_{=1}$, and it follows from the previous proposition that $f\in\F_{=1}$. Thus $\eosd\subseteq\F_{=1}$. Now, by Theorem~\ref{main}, $\eosd$ is a set of forbidden subnetworks for $\P$, and we deduce from the previous proposition that $\F_{=1}\subseteq\eosd$. 
\end{proof}
%%%%%%%%%%%%%%%%%%

Let $\P_{\leq 1}$ (resp. $\P_{\geq 1}$) be the set of networks $f$ such that each subnetwork of $f$ has at most (resp. at least) one fixed point; and let $\F_{\leq 1}$ and $\F_{\geq 1}$ be the smallest sets of forbidden subnetworks for $\P_{\leq 1}$ and $\P_{\geq 1}$, respectively. In the light of the ``proof by dichotomy'' of Robert's theorem (given by Theorem~\ref{thm:thomas}) it is tempting to try to deduce that $\eosd$ is the smallest set of forbidden subnetworks for $\P_{=1}$ from the forbidden sets $\F_{\leq 1}$ and $\F_{\geq 1}$. But this is not so simple. Indeed, $\F_{\leq 1}\cup\F_{\geq 1}$ is clearly a set of forbidden subnetworks for $\P_{\leq 1}\cap \P_{\geq 1}=\P_{=1}$, thus $\eosd\subseteq \F_{\leq 1}\cup\F_{\geq 1}$, but the inclusion is strict: A lot of networks critical for $\neg\P_{\leq 1}$ or $\neg\P_{\geq 1}$ are not critical for $\neg\P_{=1}$ (because any network that is critical for $\neg\P_{\leq 1}$ (resp. $\neg\P_{\geq 1}$) and that contains a subnetworks with no (resp. multiple) fixed point is not critical for $\neg\P_{=1}$). Examples are given below. 

\medskip
However, in Section~\ref{sec:non-exp}, we will see that, if we consider the class of non-expansive networks, then $\esd=\F_{\leq 1}$ and $\osd=\F_{\geq 1}$, so that the the equality $\eosd=\F_{\leq 1}\cup\F_{\geq 1}$ holds. Also, in Section \ref{sec:and}, we will se that $\esd=\F_{\leq 1}$ for another class of networks (the conjunctive networks), and we will leave the equality $\osd=\F_{\geq 1}$ has an open problem for this class. 

%%%%%%%%%%%%%%%%%%
\begin{remark}
$f$ is critical for $\neg\P_{\leq 1}$ if and only if $f$ has at least two fixed points, and every strict subnetwork of $f$ has at most one fixed points. In other words, {\em $\F_{\leq 1}$ is the set of 2-critical networks}. And similarly, {\em $\F_{\geq 1}$ is the set of 0-critical networks}. 
\end{remark}
%%%%%%%%%%%%%%%%%%

Among network properties closed for subnetworks, $\P_{=1}$, $\P_{\leq 1}$ and $\P_{\geq 1}$ are not ``very strong", and this is why it is interesting to characterize them in terms of forbidden subnetworks. By opposition, closed property as $\P_{>1}$ (every subnetwork has at least two fixed points) or $\P_{<1}$ (every subnetwork has no fixed points) are not interesting. To see this, consider the two one-dimensional constant networks $\Zero(x)=0$ and $\One(x)=1$. Clearly $\Zero$ and $\One$ have a unique fixed point and are thus critical for $\P_{>1}$ or $\P_{<1}$. Consequently, $\Zero$ and $\One$ are in the smallest forbidden set of subnetworks for $\P_{>1}$ and $\P_{<1}$. But it is easy to see that networks without $\Zero$ or $\One$ as subnetwork are (exactly) networks $f$ such that $\f$ is a constant, and restrict our attention to this type of networks is not interesting. Actually, even if only $\Zero$ or only $\One$ is forbidden, the resulting networks are too particular to be interesting. In other words: {\em Interesting closed properties must be satisfied by $\Zero$ and $\One$}. An interesting property different from $\P_{=1}$, $\P_{\leq 1}$ and $\P_{\geq 1}$, is for example ``each subnetwork has an asynchronous state graph which describes a strong convergence toward a unique fixed point''. Hence, it would be interesting to characterize the set of forbidden subnetworks for this property. 

%%%%%%%%%%%%%%%%%%
\begin{example} 
The following network $f$ is 2-critical ({\em i.e.} in $\F_{\leq 1}$) but not even (double arrows indicate cycles of length two):
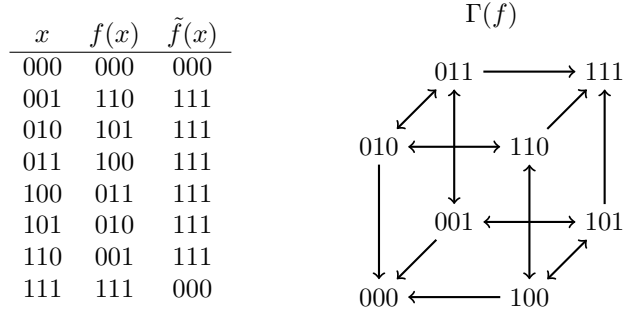
\[
\begin{array}{ccc}
x & f(x) & \f(x)\\\hline
000 & 000 & 000\\
001 & 110 & 111\\
010 & 101 & 111\\
011 & 100 & 111\\
100 & 011 & 111\\
101 & 010 & 111\\
110 & 001 & 111\\
111 & 111 & 000
\end{array}
\qquad\qquad
\begin{array}{c}
\Gamma(f)\\[3mm]
\begin{tikzpicture}
\node (000) at (-1.5,-1.5){$000$};
\node (001) at (-0.5,-0.5){$001$};
\node (010) at (-1.5,0.5){$010$};
\node (011) at (-0.5,1.5){$011$};
\node (100) at (0.5,-1.5){$100$};
\node (101) at (1.5,-0.5){$101$};
\node (110) at (0.5,0.5){$110$};
\node (111) at (1.5,1.5){$111$};
\path[thick,->]
(001) edge[<->] (101)
(001) edge (000)
(010) edge (000)
(010) edge[<->] (011)
(011) edge[<->] (001)
(100) edge (000)
(110) edge[<->] (100)
(100) edge[<->] (101)
(110) edge[<->] (010)
(011) edge (111)
(101) edge (111)
(110) edge (111)
;
\end{tikzpicture}
\end{array}
\]
The following network $f$ is 0-critical ({\em i.e.} in $\F_{\geq 1}$) but not odd:
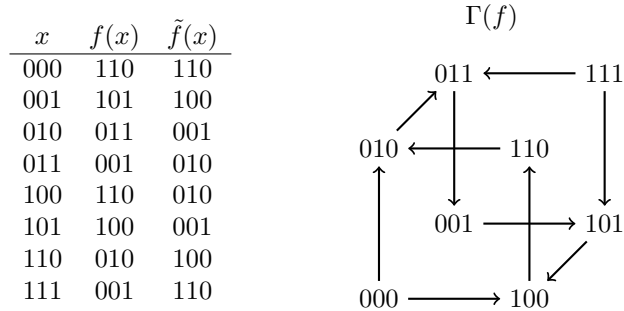
\[
\begin{array}{ccc}
x & f(x) & \f(x)\\\hline
000 & 110 & 110\\
001 & 101 & 100\\
010 & 011 & 001\\
011 & 001 & 010\\
100 & 110 & 010\\
101 & 100 & 001\\
110 & 010 & 100\\
111 & 001 & 110
\end{array}
\qquad\qquad
\begin{array}{c}
\Gamma(f)\\[3mm]
\begin{tikzpicture}
\node (000) at (-1.5,-1.5){$000$};
\node (001) at (-0.5,-0.5){$001$};
\node (010) at (-1.5,0.5){$010$};
\node (011) at (-0.5,1.5){$011$};
\node (100) at (0.5,-1.5){$100$};
\node (101) at (1.5,-0.5){$101$};
\node (110) at (0.5,0.5){$110$};
\node (111) at (1.5,1.5){$111$};
\path[thick,->]
(000) edge (100)
(000) edge (010)
(001) edge (101)
(010) edge (011)
(011) edge (001)
(100) edge (110)
(101) edge (100)
(110) edge (010)
(111) edge (011)
(111) edge (101)
;
\end{tikzpicture}
\end{array}
\]
\end{example}
%%%%%%%%%%%%%%%%%%

%%%%%%%%%%%%%%%%%%%%%%%%%%%%%%%%%%%%%%%%%%%%%%%%%%%%%%%%%%%%%%%%%%%%%%
%%%%%%%%%%%%%%%%%%%%%%%%%%%%%%%%%%%%%%%%%%%%%%%%%%%%%%%%%%%%%%%%%%%%%%
\section{Circular networks and non expansive networks}\label{sec:cycle}
%%%%%%%%%%%%%%%%%%%%%%%%%%%%%%%%%%%%%%%%%%%%%%%%%%%%%%%%%%%%%%%%%%%%%%
%%%%%%%%%%%%%%%%%%%%%%%%%%%%%%%%%%%%%%%%%%%%%%%%%%%%%%%%%%%%%%%%%%%%%%

%All the theorems stated in section~\ref{} highlight the importance of cycle in the interaction graph of a Boolean network. It is then not surprising that network whose these class of network has been deeply studied; see {\cite{RMC03}} and {\cite{DNS10}} for instance. 

A {\bf positive-circular} (resp. {\bf negative-circular}) network is a network $f$ such that $G(f)$ is a positive (resp. negative) cycle. Positive- and negative-circular networks have been widely studied  ({\it e.g.} \cite{RMC03,DNS12}) because that are the ``simplest non simple networks'' in the sense that they are the most simple networks (from a structural point of view) that do not describe a convergence toward a unique fixed point.    

\medskip
In this section, we show that positive-circular (resp negative-circular) networks are even-self-dual (resp. odd-self-dual), and we prove that the converse holds for non-expansive networks (cf. Theorem~\ref{thm:circular-eosd} below). In this way, even- and odd-self-dual network may be seen as generalization of circular networks. 

\medskip
Suppose that $f$ is a circular network. Let $\s$ be the permutation of $V$ that maps every vertex $i$ to the vertex $\sigma(i)$ preceding $i$ in $G(f)$. For each $x\in\B^V$, let us denote by $\s x$ the point of $\B^V$ such that $(\s x)_i=x_{\s(i)}$ for all $i\in V$. Let $s\in\B^V$ be the such that for all $i\in V$, $s_i=0$ if the arc from $\sigma(i)$ to $i$ is positive and $s_i=1$ otherwise. Then, for all $x\in\B^V$, we have 
\[
f(x)=\s x\o s.   
\]
We call $\s$ the {\bf permutation of} $f$ and $s$ the {\bf constant of} $f$. Since $G(f)$ only depends on $f$, since the couple $(\s,s)$ only depends on $G(f)$ and since $f$ only depends on this couple, these three objects share the same information. In particular the sign of $G(f)$ is ``contained'' in $s$: It is positive if $\d{s}$ is even, and negative if $\d{s}$ is odd.  

%%%%%%%%%%%%%%%%%%
\begin{theorem}\label{thm:circular-eosd}
~
\begin{enumerate}
\item
$f$ is positive-circular if and only if $f$ is even-self-dual and non-expansive.
\item
$f$ is negative-circular if and only if $f$ is odd-self-dual and non-expansive.
\end{enumerate}
\end{theorem}
%%%%%%%%%%%%%%%%%%

We will use the following lemma several times.

%%%%%%%%%%%%%%%%%%
\begin{lemma}\label{lem:Bf}
Let $f$ be networks on $V$ and $I\subseteq V$. Let $f'$ be the network on~$V$ defined by $f'(x)=f(x)\o e_I$ for all $x\in\B^V$. We have the following properties.
\begin{enumerate}
\item
If $f$ is non-expansive, then $f'$ is non-expansive.
\item
If $f$ is self-dual, then $f'$ is self-dual.
\item
If $f$ is even or odd, then $f'$ is even or odd.
\item
$|Gf'(x)|=|Gf(x)|$ for all $x\in\B^V$. 
\end{enumerate}
\end{lemma}
%%%%%%%%%%%%%%%%%%

%%%%%%%%%%%%%%%%%%
\begin{proof}
Suppose that $f$ is non-expansive, and let $x,y\in\B^V$. Then 
\[
d(f'(x),f'(y))=d(f(x)\o e_I,f(y)\o e_I)=d(f(x),f(y))\leq d(x,y).
\]
thus $f'$ is non-expansive. 

\medskip
If $f$ is self-dual then 
\[
f'(x\o 1)=f(x\o 1)\o e_I=f(x)\o 1\o e_I=f'(x)\o 1
\]
thus $f'$ is self-dual. 

\medskip
Suppose that $f$ is even. For all $x\in\B^V$, we have $\tilde f'(x)=f'(x)\o x=f(x)\o e_I\o x=\f(x)\o e_I$, and since $\f(x)$ is even, we deduce that $\tilde f'(x)$ and $|I|$ have the same parity. Thus all the points of $\tilde f'(\B^V)$ have the parity of $|I|$. Suppose that $|I|$ is even (resp. odd), and let $z\in\B^V$ be an even (resp. odd). Then $z\o e_I$ is even, thus there exists $x\in\B^V$ such that $\f(x)=z\o e_I$, so $\tilde f'(x)=f'(x)\o x=\f(x)\o e_I=z$. Thus every even (resp. odd) point of $\B^V$ is in $\tilde f'(\B^V)$. Thus $f'$ is even if $|I|$ is even, and $f'$ is odd otherwise. The proof is similar if $f$ is odd.  

\medskip
For all $i,j\in V$ and $x\in\B^V$, 
\begin{multline*}
|f'_{ij}(x)|=f'_i(x)\o f'_i(x\o e_j)
=f_i(x)\o e_I\o f_i(x\o e_j)\o e_I\\[1mm]
=f_i(x)\o f_i(x\o e_j)
=|f_{ij}(x)|
\end{multline*}
and the last point follows. 
\end{proof}
%%%%%%%%%%%%%%%%%%

%%%%%%%%%%%%%%%%%%
\begin{proof}[Proof of Theorem~\ref{thm:circular-eosd}]
{\it (Direction $\Rightarrow$)} Let $f$ be a circular with permutation $\s$ and constant $s$. For all $x,y\in\B^V$, we have 
\[
d(f(x),f(y))=\d{\s x\o s\o\s y\o s}=\d{\s x\o\s y}=\d{x\o y}=d(x,y).
\]
thus $f$ is non expansive. Also,  
\[
f(x\o 1)=\s(x\o 1)\o s=\s x\o 1\o s=f(x)\o 1
\] 
thus $f$ is self-dual. We now prove that $f$ is even (resp. odd) if $G(f)$ is positive
(resp. negative). We have $\f(x)=x\o \s x\o s$ so the parity
of $\f(x)$ is the parity of
$\d{x}+\d{\s x}+\d{s}$. Since
$\d{x}=\d{\s x}$, we deduce that the parity of $\f(x)$
is the parity of $\d{s}$. So if $G(f)$ is positive (resp. negative) then the
image of $\f$ only contains even (resp. odd) points. It remains to prove
that if $G(f)$ is positive (resp. negative) then each even (resp. odd) point is in
the image of $\f$. Suppose that $G(f)$ is positive (resp. negative), and let
$z$ be an even (resp. odd) point of $\B^V$. Let $n=|V|$ and let $i_1,i_2,\dots i_n$ be the vertices of $G(f)$ given in the order, so that $\s(i_1)=i_n$ and $\s(i_{k+1})=i_k$ for $1\leq k<n$. Let $x$ be the point of $\B^V$ whose components $x_{i_k}$ are recursively defined as follows, with $k$ decreasing from $n$ to $1$:
\[
x_{i_n}=z_{i_1},\qquad x_{i_k}=z_{i_{k+1}}\o s_{i_{k+1}}\o x_{i_{k+1}}\qquad 1\leq k<n. 
\]
Let us prove that $\f(x)=z$. For every $1<k\leq n$, we have
\begin{multline*}
\f_{i_k}(x)=f_{i_k}(x)\o x_{i_k}=x_{\sigma(i_k)}\o s_{i_k}\o x_{i_k}\\[1mm]=x_{i_{k-1}}\o s_{i_k}\o x_{i_k}=(z_{i_k}\o s_{i_k}\o x_{i_k})\o s_{i_k}\o x_{i_k}=z_{i_k}.
\end{multline*}
It remains to prove that $\f_{i_1}(x)=z_{i_n}$. By the definition of $x$, we have 
\[
\begin{array}{rcl}
x_{i_1}&=&(z_{i_2}\o s_{i_2})\o x_{i_2}\\[1mm]
&=&(z_{i_2}\o s_{i_2})\o (z_{i_3}\o s_{i_3})\o x_{i_3}\\[1mm]
&\vdots&\\[1mm]
&=&(z_{i_2}\o s_{i_2})\o (z_{i_3}\o s_{i_3})\o\cdots\o (z_{i_n}\o s_{i_n})\o z_{i_1}\\[1mm]
&=&(z_{i_2}\o z_{i_3}\o\cdots\o z_{i_n}\o z_{i_1})\o(s_{i_2}\o s_{i_3}\o\cdots\o s_{i_n}).
\end{array}
\]
So $\d{z}$ and $\d{x_{i_1}\o s_{i_2}\o s_{i_3}\o\cdots\o s_{i_n}}$ have the same parity, and
since $\d{z}$ and $\d{s}$ have the same parity, we deduce that $x_{i_1}=s_{i_1}$. Thus
\[
\f_{i_1}(x)=f_{i_1}(x)\o x_{i_1}=x_{i_n}\o s_{i_1}\o x_{i_1}=z_{i_1}\o s_{i_1}\o s_{i_1}=z_{i_1}
\]
and it follows that $\f(x)=z$. So $f$ is even (resp. odd).

%%%%%%%%%%%%%%%%%%%
\medskip
{\it (Direction $\Leftarrow$)} We first prove the following property:
\begin{enumerate}
\item[{\bf (1)}] Suppose that $f$ is odd-self-dual and non-expansive, and suppose that $f(x)=x\o e_i$ for some $x\in\B^V$ and $i\in V$. Then the in-degree of $i$ in $Gf(x)$ is at most~one. 
\end{enumerate}
Let $x^1=x$, and for all $k\in\mathbb{N}$, let $x^{k+1}=f(x^k)$. Let $n=|V|$, and for all $1\leq p\leq n$, let us say that a sequence $i_1,i_2,\dots,i_p$ is {\em good} if it is a sequence of $p$ distinct vertices in $V$ such that 
\[
f(x^k)=x^k\o e_{i_k}\qquad 1\leq k\leq p.
\]
Let us prove the following property: 
\begin{enumerate}
\item[($*$)] For all $1\leq p\leq n$, there exists a good sequence $i_1,i_2,\dots, i_p$. 
\end{enumerate}
Since $f(x)=x\o e_i$, this is true for $p=1$. So suppose that $1<p\leq n$ and that there exists a good sequence $i_1,i_2,\dots, i_{p-1}$. Since $x^p=f(x^{p-1})=x^{p-1}\o e_{i_{p-1}}$ we have $d(x^p,x^{p-1})=1$, and since $f$ is non-expansive, we deduce that  
\[
d(f(x^p),x^p)=d(f(x^p),f(x^{p-1}))\leq d(x^p,x^{p-1})=1.
\]
Since $f$ is odd, $f$ has no fixed point, thus $d(f(x^p),x^p)=1$, {\em i.e.} there exists an element of $V$, that we denote by $i_p$, such that $f(x^p)=x^p\o e_{i_p}$. To complete the induction step, it remains to prove that $i_p\neq i_1,i_2\dots,i_{p-1}$. Suppose, for a contradiction, that $i_p=i_k$ with $1\leq k<p$. Then $\f(x^p)=\f(x^k)=e_{i_k}$. Since $f$ is self dual $\f(x^p\o 1)=e_{i_k}$. Thus $x^p$, $x^k$ and $x^p\o 1$ are elements of $\f^{-1}(e_{i_k})$. Since $f$ is odd-self-dual, $\f^{-1}(e_{i_k})$ contains exactly two elements. Thus $x^p=x^k$ or $x^p\o 1=x^k$, and this is not possible since
\[
\begin{array}{rcl}
x^p&=&x^{p-1}\o e_{i_{p-1}}\\[1mm]
&=&x^{p-2}\o e_{i_{p-2}}\o e_{i_{p-1}}\\[1mm]
&\vdots&\\[1mm]
&=&x^k\o e_{i_k}\o e_{i_{k+1}}\o\cdots\o e_{i_{p-2}}\o e_{i_{p-1}}.
\end{array}
\]
This prove the induction step and ($*$) follows. So let $i_1,i_2,\dots,i_n$ be a good sequence. Since $x=x^1$ and $f(x)=x\o e_i$, we have $i=i_1$. To prove {\bf (1)}, we will prove that if $Gf(x)$ has an arc from $i_k$ to $i$, then $k=n$. So let $1\leq k\leq n$, and suppose that $Gf(x)$ contains an arc from $i_k$ to $i$. Since $f$ is non-expansive, 
\[
f(x\o e_{i_k})=f(x)\o e_i=x\o e_i\o e_i=x=(x\o e_{i_{k}})\o e_{i_{k}} 
\]
Thus $\f(x\o e_{i_k})=\f(x^k)=e_{i_k}$. Since $f$ is self-dual $\f(x\o e_{i_k}\o 1)=e_{i_k}$. Thus $x\o e_{i_k}$, $x^k$ and $x\o e_{i_k}\o 1$ are elements of $\f^{-1}(e_{i_k})$, and as previously, we deduce that $x^k=x\o e_{i_k}$ or $x^k=x\o e_{i_k}\o 1$. Thus $k>1$, and since 
\[
x^k=x\o e_{i_1}\o e_{i_2}\o\cdots \o e_{i_{k-1}}
\]
we have $x^k\neq x\o e_{i_k}$. Thus $x^k=x\o e_{i_k}\o 1$ so $e_{i_k}\o 1=e_{i_1}\o e_{i_2}\o\cdots \o e_{i_{k-1}}$. If $k<n$ then $(e_{i_k}\o 1)_{i_n}=1$ and $(e_{i_1}\o e_{i_2}\o\cdots \o e_{i_{k-1}})_{i_n}=0$, a contradiction. Thus $k=n$ and {\bf (1)} is proved.    
\begin{enumerate}
\item[{\bf (2)}] Suppose that $f$ even-self-dual and non-expansive and suppose that $f(x)=x$ for some $x\in\B^V$. Then $Gf(x)$ is a disjoint union of cycles.
\end{enumerate}
Let $i\in V$. If $f(x\o e_i)=x$ then $\f(x\o e_i)=e_i$ and this is not possible since $f$  is even. Since $f$ is non-expansive, we deduce that there exists $j\in V$ such that $f(x\o e_i)=f(x)\o e_j$. Then $j$ is the unique out-neighbor of $i$ in $Gf(x)$. Thus we have prove the following:
\begin{enumerate}
\item[($*$)] Each vertex of $Gf(x)$ has exactly one out-neighbor. 
\end{enumerate}
Let $i\in V$, and let $h$ be the network on $V$ defined by $f'(y)=f(y)\o e_i$ for all $y\in\B^V$. Since $f(x)=x$, we have $f'(x)=x\o e_i$, thus according to Lemma~\ref{lem:Bf}, $f'$ is odd-self-dual and non-expansive. So according to {\bf (1)}, $i$ has at most one in-neighbor in $Gf'(x)$, and by Lemma~\ref{lem:Bf}, $i$ has at most one in-neighbor in $Gf(x)$. Thus each vertex of $Gf(x)$ has at most one in-neighbor, and using ($*$) we deduce that each vertex of $Gf(x)$ has exactly one in-neighbor. Consequently, $Gf(x)$ is a disjoint union of cycles. This proves {\bf (2)}. 
\begin{enumerate}
\item[{\bf (3)}] Suppose that $f$ is even- or odd-self-dual and non-expansive. Then $Gf(x)$ is a disjoint union of cycles for all $x\in\B^V$.
\end{enumerate}
Let $x\in\B^V$, and let $f'$ be the network on $V$ defined by $f'(y)=f(y)\o\f(x)$ for all $y\in\B^V$. Then $f'(x)=f(x)\o \f(x)=x\o \f(x)\o\f(x)=x$ and we deduce from Lemma~\ref{lem:Bf} that $f'$ is even-self-dual and non-expansive. Thus, following~{\bf(2)}, $Gf'(x)$ is a disjoint union of cycles, and we deduce from Lemma~\ref{lem:Bf} that $Gf(x)$ is a disjoint union of cycles. This proves {\bf (3)}. 
\begin{enumerate}
\item[{\bf (4)}] 
Suppose that $f$ is even- or odd-self-dual and non-expansive. Then $Gf(x)=G(f)$ for all $x\in\B^V$.
\end{enumerate}
Let $x\in\B^V$ and $i,k,l\in V$. Suppose that 
\[
f_{lk}(x)=s\neq 0
\qquad\text{and}\qquad
f_{lk}(x\o e_i)\neq s.
\]
Since $f_{lk}(x)=f_{lk}(x\o e_k)$, we have $k\neq i$, and since, by {\bf (3)}, each vertex of $Gf(x)$ has a unique in-neighbor, we have $f_l(x)=f_l(x\o e_i)$. Suppose that $x_k=0$. Then 
\[
f_{lk}(x\o e_i)=f_l(x\o e_i\o e_k)-f_l(x\o e_i)=f_l(x\o e_i\o e_k)-f_l(x)\neq s
\]
and $f_{lk}(x)=f_l(x\o e_k)-f_l(x)=s$. Thus $f_l(x\o e_i\o e_k)\neq f_l(x\o e_k)$, that is, $f_{li}(x\o e_k)\neq 0$. Thus $Gf(x\o e_k)$ contains both an arc from $k$ to $l$ and from $i$ to $l$. Since $i\neq k$, $l$ has at least two in-neighbor in $Gf(x\o e_k)$, and this contradicts {\bf (3)}. If $x_k=1$, we obtain a contradiction with similar arguments. Thus:
\[
\forall x\in\B^V,~\forall i,k,l\in V, \qquad f_{lk}(x)\neq 0~\Rightarrow 
f_{lk}(x)=f_{lk}(x\o e_i)
\]
We deduce that $Gf(x)$ is a subgraph of $Gf(x\o e_i)$ and that $Gf(x\o e_i)$ is a subgraph of $Gf((x\o e_i)\o e_i)=Gf(x)$. Thus $G(x)=G(x\o e_i)$ for all $x\in\B^V$ and $i\in V$, and as an immediate consequence, $Gf(x)=Gf(y)$ for all $x,y\in\B^V$. This proves {\bf (4)}. 
\begin{enumerate}
\item[{\bf (5)}] If $f$ is even-self-dual and non-expansive, then $G(f)$ is a positive cycle.
\end{enumerate}
Indeed, following {\bf (3)} and {\bf (4)}, $G(f)$ is a disjoint union of cycles and since $f$ if even-self-dual, $f$ has exactly $2$ fixed points. Thus $G(f)$ has only positive cycles (otherwise $f$ would have no fixed point, according to Theorem~\ref{thm:ara}). And since if $G(f)$ is a union of $p\geq 1$ disjoint positive cycle then $f$ has $2^p$ fixed points, we deduce that $G(f)$ is a positive cycle. 
\begin{enumerate}
\item[{\bf (6)}] If $f$ is odd-self-dual and non-expansive, then $G(f)$ is a negative cycle.
\end{enumerate}
Let $i\in V$ and let $f'$ be the network on $V$ defined by $f'(x)=f(x)\o e_i$ for all $x\in\B^V$. By Lemma~\ref{lem:Bf}, $f'$ is even-self-dual and non-expansive. Thus according to {\bf (5)}, $G(f')$ is a cycle. From Lemma~\ref{lem:Bf}, we deduce that $G(f)$ is a cycle too. Since $f$ is odd, it has no fixed point, and we deduce that $G(f)$ is a negative cycle.
\end{proof}
%%%%%%%%%%%%%%%%%%

As an immediate consequence of this theorem and Corollary~\ref{cor2} we obtain the following:

%%%%%%%%%%%%%%%%%%
\begin{corollary}\label{cor:nonexp1}
If $f$ is non-expansive, then every subnetwork of $f$ has a unique fixed point if and only if $f$ has no circular subnetwork. 
\end{corollary}
%%%%%%%%%%%%%%%%%%

%%%%%%%%%%%%%%%%%%
\begin{remark}
It is easy to check that critical even-self-dual (resp. odd-self-dual) network with at most three components are circular. Below is an example of critical even-self-dual network with four components which is not circular.   
\end{remark}
%%%%%%%%%%%%%%%%%%
	
%%%%%%%%%%%%%%%%%%
\begin{example}
The following network $f$ on $\{1,2,3,4\}$ is a critical even-self-dual network which is not circular  (and which is expansive, since $d(f(0),f(e_i))\geq 2$ for $i=1,2,3,4$). Note that the subnetwork $f^{40}$ is the three-dimensional network considered in Examples~\ref{ex1}, \ref{ex2} and \ref{ex3}. 
\[
\begin{array}{l}
f_1(x)=(\overline{x_2}\land x_3\land \overline{x_4})\lor
((\overline{x_2}\lor x_3)\land x_4)\\[1mm]
f_2(x)=(\overline{x_3}\land x_1\land \overline{x_4})\lor
((\overline{x_3}\lor x_1)\land x_4)\\[1mm]
f_3(x)=(\overline{x_1}\land x_2\land \overline{x_4})\lor
((\overline{x_1}\lor x_2)\land x_4)\\[1mm]
f_4(x)=(x_2\land x_3\land \overline{x_1})\lor
((x_2\lor x_3)\land x_1)
\end{array}
\]
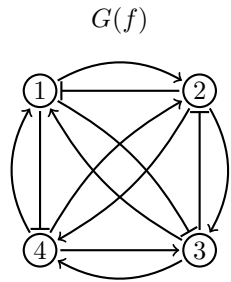
\[
\begin{array}{ccc}
\begin{array}{ccc}
x & f(x) & \f(x)\\\hline
0000 & 0000 & 0000\\
0010 & 1000 & 1010\\
0100 & 0010 & 0110\\
0110 & 0011 & 0101\\
1000 & 0100 & 1100\\
1010 & 1001 & 0011\\
1100 & 0101 & 1001\\
1110 & 0001 & 1111\\
0001 & 1110 & 1111\\
0011 & 1010 & 1001\\
0101 & 0110 & 0011\\
0111 & 1011 & 1100\\
1001 & 1100 & 0101\\
1011 & 1101 & 0110\\
1101 & 0111 & 1010\\
1111 & 1111 & 0000
\end{array}
&\qquad&
\begin{array}{c}
\begin{tikzpicture} 
\node[outer sep=1.5,inner sep=1.5] (0) at (90:2){$G(f)$};
\node[outer sep=1.5,inner sep=1.5,circle,draw,thick] (1) at (135:1.5){$1$};
\node[outer sep=1.5,inner sep=1.5,circle,draw,thick] (2) at (45:1.5){$2$};
\node[outer sep=1.5,inner sep=1.5,circle,draw,thick] (3) at (-45:1.5){$3$};
\node[outer sep=1.5,inner sep=1.5,circle,draw,thick] (4) at (-135:1.5){$4$};
\path[thick]
(2) edge[bend left=00,-|] (1)
(3) edge[bend left=15,->] (1)
(4) edge[bend left=30,->] (1)
(1) edge[bend left=30,->] (2)
(3) edge[bend left=00,-|] (2)
(4) edge[bend left=15,->] (2)
(1) edge[bend left=15,-|] (3)
(2) edge[bend left=30,->] (3)
(4) edge[bend left=00,->] (3)
(1) edge[bend left=00,-|] (4)
(2) edge[bend left=15,->] (4)
(3) edge[bend left=30,->] (4)
;
\end{tikzpicture}
\end{array}
\end{array}
\]
\end{example}
%%%%%%%%%%%%%%%%%%

%%%%%%%%%%%%%%%%%%%%%%%%%%%%%%%%%%%%%%%%%%%%%%%%%%%%%%%%%%%%%%%%%%%%%%
%%%%%%%%%%%%%%%%%%%%%%%%%%%%%%%%%%%%%%%%%%%%%%%%%%%%%%%%%%%%%%%%%%%%%%
\section{Non-expansive networks}\label{sec:non-exp}
%%%%%%%%%%%%%%%%%%%%%%%%%%%%%%%%%%%%%%%%%%%%%%%%%%%%%%%%%%%%%%%%%%%%%%
%%%%%%%%%%%%%%%%%%%%%%%%%%%%%%%%%%%%%%%%%%%%%%%%%%%%%%%%%%%%%%%%%%%%%%

As we have seen in the preceding section, a positive-circular (resp. negative-circular) network $f$ is non-expansive, and it is easy to see that such a network is also 2-critical (resp. 0-critical). The following theorem, the main result of this section, asserts that the converse is true. 

%%%%%%%%%%%%%%%%%%
\begin{theorem}\label{thm:nonexp}~
\begin{enumerate}
\item
$f$ is positive-circular if and only if $f$ is 2-critical and non-expansive. 
\item
$f$ is negative-circular if and only if $f$ is 0-critical and non-expansive. 
\end{enumerate}
\end{theorem}
%%%%%%%%%%%%%%%%%%

\noindent
Even if the two points of this theorem seem similar (symmetrical), their proofs are very different. The proof of the first is rather direct and uses Theorem~\ref{main} and a part of Theorem~\ref{thm:circular-eosd} (non-expensive even-self-dual networks are positive-circular). The proof of the second points is independent of previous results. It consists in visiting each point of $\B^V$ in a very special order. In both cases, the following lemma will be useful.

%%%%%%%%%%%%%%%%%%
\begin{lemma}\label{lem:Hf}
Let $f$ be networks on $V$ and $I\subseteq V$. Let $f'$ be the network on~$V$ defined by $f'(x)=f(x\o e_I)\o e_I$ for all $x\in\B^V$. We have the following properties.
\begin{enumerate}
\item
If $f$ is non-expansive, then $f'$ is non-expansive.
\item
If $f$ is 2-critical, then $f'$ is 2-critical.
\item
If $f$ is 0-critical, then $f'$ is 0-critical.
\item
$|G(f')|=|G(f)|$. 
\end{enumerate}
\end{lemma}
%%%%%%%%%%%%%%%%%%

%%%%%%%%%%%%%%%%%%
\begin{proof}
Suppose that $f$ is non-expansive, and let $x,y\in\B^V$. Then 
\begin{multline*}
d(f'(x),f'(y))=d(f(x\o e_I)\o e_I,f(y\o e_I)\o e_I)\\[1mm]
=d(f(x\o e_I),f(y\o e_I))\leq d(x\o e_I,y\o e_I)=d(x,y).
\end{multline*}
thus $f'$ is non-expansive. 

\medskip
Let $J$ be a non-empty subset of $V$ and let $h$ be the subnetwork of $f$ induced by $z\in\B^{V\setminus J}$. Let $h'$ be the network on $J$ defined by $h'(y)=h(y\o e_{I\cap J})\o e_{I\cap J}$ for all $y\in\B^J$. Let $x\in\B^V$ be such that $x_{-J}=z\o e_{I\setminus J}$. We have 
\[
h'(x|_J)=g(x|_J\o e_{I\cap J})\o e_{I\cap J}=h((x\o e_I)|_J)\o e_{I\cap J}
\]
Since 
\[
(x\o e_I)_{-J}=x_{-J}\o e_{I\setminus I}=z\o e_{I\setminus J}\o e_{I\setminus I}=z
\]
we have 
\[
h((x\o e_I)|_J)\o e_{I\cap J}=f(x\o e_I)|_J\o e_{I\cap J}=(f(x\o e_I)\o e_I)|_J=f'(x)|_J.
\]
Thus $h'(x|_J)=f'(x)|_J$ for all $x\in\B^V$ be such that $x_{-J}=z\o e_{I\setminus J}$, {\em i.e.} $h'$ is the subnetwork of $f'$ induced by $z\o e_{I\setminus J}$. Since it is clear that $h$ and $h'$ have the same number of fixed points, {\em 2.} and {\em 3.} are proved. 

\medskip
For all $i,j\in V$ and $x\in\B^V$, 
\begin{multline*}
|f'_{ij}(x)|=f'_i(x)\o f'_i(x\o e_j)
=f_i(x\o e_I)\o e_I\o f_i(x\o e_I\o e_j)\o e_I\\[1mm]
=f_i(x\o e_I)\o f_i(x\o e_I\o e_j)
=|f_{ij}(x\o e_I)|
\end{multline*}
thus $|Gf'(x)|=|Gf(x\o e_I)|$ and {\em 4.} follows.
\end{proof}
%%%%%%%%%%%%%%%%%%

%%%%%%%%%%%%%%%%%%
\begin{proof}[Proof of Theorem~\ref{thm:nonexp}]
{\it (Direction $\Rightarrow$ of 1. and 2.)} Suppose that $f$ is positive-circular (resp. negative-circular). According to Theorem~\ref{thm:circular-eosd}, $f$ is non-expansive, and according to the same theorem, it is even-self-dual (resp. odd-self-dual), thus it has two (resp. no) fixed points. If $h$ is a strict subnetwork of $f$, then $G(h)$ is a strict subgraph of $G(f)$, thus it is acyclic, and by Robert's theorem, $h$ has a unique fixed point. Thus $f$ is 2-critical (resp. 0-critical).  

\medskip
%%%%%%%%%%%%%%%%%%%%%%%%%%%%%%%%%%%%%%%%%%%%%%%%%%%%%%
{\it (Direction $\Leftarrow$ of 1.)}
%%%%%%%%%%%%%%%%%%%%%%%%%%%%%%%%%%%%%%%%%%%%%%%%%%%%%%
We first need the following property:
\begin{enumerate}
\item[{\bf (1)}]%1
Suppose that $f$ is non-expansive. If $f(0)=0$ and $f(1)=1$ then $\d{x}=\d{f(x)}$ for all $x\in\B^V$. 
\end{enumerate}
Indeed, under these hypothesis, 
\[
\d{f(x)}=d(0,f(x))=d(f(0),f(x))\leq d(0,x)=\d{x}
\]
and
\[
|V|-\d{f(x)}=d(1,f(x))=d(f(1),f(x))\leq d(1,x)=|V|-\d{x}.
\]
thus $\d{f(x)}\geq \d{x}$ and it follows that $\d{f(x)}=\d{x}$. 

\begin{enumerate}
\item[{\bf (2)}]%1
Suppose that $f$ is non-expansive. Suppose also that $f(0)=0$ and $f(1)=1$. Let $I$ be a non-empty subset of $V$. Let $z\in\B^{V\setminus I}$ and let $h$ be the subnetwork of $f$ induced by $z$. If $h(1)=h(0)\o 1$ then $h(0)=0$. 
\end{enumerate}
Let $z^0$ and $z^1$ denotes the points of $\B^V$ such that 
\[
z^0|_I=0,\qquad z^1|_I=1,\qquad z^0_{-I}=z^1_{-I}=z.
\]
So $h(0)=f(z^0)|_I$ and $h(1)=f(z^1)|_I$. Suppose that $h(1)=h(0)\o 1$. Then 
\[
\begin{array}{rclcl}
d(f(z^0),f(z^1))&=&d(f(z^0)_{-I},f(z^1)_{-I})+d(f(z^0)|_I,f(z^1)|_I)\\[1mm]
&=&d(f(z^0)_{-I},f(z^1)_{-I})+d(h(0),h(1))\\[1mm]
&=&d(f(z^0)_{-I},f(z^1)_{-I})+d(h(0),h(0)\o 1)\\[1mm]
&=&d(f(z^0)_{-I},f(z^1)_{-I})+|I|.
\end{array}
\]
Since $f$ is non-expansive
\[
d(f(z^0),f(z^1))=d(f(z^0)_{-I},f(z^1)_{-I})+|I|\leq d(z^0,z^1)=|I|
\]
thus $d(f(z^0)_{-I},f(z^1)_{-I})=0$, that is $f(z^0)_{-I}=f(z^1)_{-I}=y$ for some $y\in\B^{V\setminus I}$. Since $f(0)=0$ and $f(1)=1$, it follows from {\bf(1)} that 
\[
\d{z}=\d{z^0}=\d{f(z^0)}=\d{f(z^0)|_I}+\d{f(z^0)_{-I}}=\d{h(0)}+\d{y}
\]
and 
\begin{multline*}
|I|+\d{z}=\d{z^1}=\d{f(z^1)}=\d{f(z^1)|_I}+\d{f(z^1)_{-I}}\\[1mm]=\d{h(1)}+\d{y}=\d{h(0)\o 1}+\d{y}=|I|-\d{h(0)}+\d{y}.
\end{multline*}
Thus 
\[
2\d{z}=\d{z^0}+\d{z^1}-|I|=2\d{y}.
\]
Hence $\d{z}=\d{y}$ and since $\d{z}=\d{h(0)}+\d{y}$, and it follows that $\d{h(0)}=0$. This prove {\bf (2)}.

\medskip
We are now in position to prove that 2-critical non-expansive networks are positive-circular. Suppose that $f$ is 2-critical and non-expansive. Let $x$ be a fixed point of $f$. Let $f'$ be the network on $V$ defined by $f'(y)=f(y\o x)\o x$ for all $y\in\B^V$. Then  $f'(0)=f(x)\o x=x\o x=0$. Furthermore, by Lemma~\ref{lem:Hf}, $f'$ is 2-critical (so $f'(1)=1$) and $f'$ is non-expansive. Suppose that $f'$ has a self-dual strict subnetwork $h$. Then following~{\bf (2)}, we have $h(0)=0$ and thus $h(1)=1$, so $f'$ is not 2-critical, a contradiction. We deduce that $f'$ has no self-dual strict subnetwork, and since it has two fixed points, we deduce from Theorem~\ref{main} that $f'$ is even-self-dual. Thus, according to Theorem~\ref{thm:circular-eosd}, $G(f')$ is a  positive cycle. It follows from Lemma~\ref{lem:Hf} that $G(f)$ is a cycle, and since $f$ has two fixed points, $G(f)$ is a positive cycle. 

\medskip
%%%%%%%%%%%%%%%%%%%%%%%%%%%%%%%%%%%%%%%%%%%%%%%%%%%%%%
{\it (Direction $\Leftarrow$ of point 2.)}
%%%%%%%%%%%%%%%%%%%%%%%%%%%%%%%%%%%%%%%%%%%%%%%%%%%%%%
We begin with the following fact. 
\begin{enumerate}%1
\item[{\bf (3)}] 
If $f$ is non-expansive and 0-critical, then for all $i\in V$ there exists $x,y\in \B^V$ with $x_i\neq y_i$ such that $\f(x)=\f(y)=e_i$. 
\end{enumerate}
Let $i\in V$ and $\a\in\B$. Since $f$ is 0-critical, the immediate subnetwork $f^{i\a}$ has at least one fixed point. Thus there exists $x\in\B^V$ with $x_i=\a$ such that $f(x)_{-i}=f^{i\a}(x_{-i})=x_{-i}$. Hence, $f(x)=x$ or $f(x)=x\o e_i$, and since $f$ has no fixed point, we deduce that $f(x)=x\o e_i$. Thus $\f(x)=e_i$, and {\bf(3)} follows.
\begin{enumerate}
\item[{\bf (4)}] 
If $f$ is non-expansive and 0-critical then for all $i\in V$ and $x,y\in\B^V$:
\[
\f(x)=\f(y)=e_i~\mathrm{and}~x_i\neq y_i\quad\Rightarrow\quad x=y\o 1.
\]
\end{enumerate}
Suppose that $\f(x)=\f(y)=e_i$ and $x_i\neq y_i$. Suppose that there exists $j$ such that $x_j=y_j=\a$. Then $f^{j\a}(x_{-j})=x_{-j}\o e_i$ and $f^{j\a}(y_{-j})=y_{-j}\o e_i$. Since $f$ is 0-critical, $f^{j\a}$ has a fixed point $z$. If $x_i=z_i$ then 
\[
d(f^{j\a}(x_{-j}),f^{j\a}(z))=d(x_{-j}\o e_i,z)=d(x_{-j},z)+1,
\]
a contradiction with the fact that $f^{j\a}$ is non-expansive. Otherwise $y_i=z_i$~so  
\[
d(f^{j\a}(y_{-j}),f^{j\a}(z))=d(y_{-j}\o e_i,z)=d(y_{-j},z)+1
\]
and we obtain the same contradiction. Consequently, there is no $j$ such that $x_j=y_j$. So $x=y\o 1$ and {\bf(4)} is proved. 
\begin{enumerate}
\item[{\bf (5)}] 
Suppose that every 0-critical non-expansive network $f$ such that $f(0)=e_i$ for some $i\in V$ is negative circular. Then every 0-critical non-expansive network is negative circular. 
\end{enumerate}
Indeed, let $f$ be 0-critical and non-expansive. By {\bf (3)} there exists $i\in V$ and $x\in \B^V$ such that $f(x)=x\o e_i$. Let $f'$ be the network on $V$ defined by $f'(y)=f(y\o x)\o x$ for all $y\in\B^V$. By Lemma~\ref{lem:Hf}, $f'$ is 0-critical and non-expansive. Furthermore, $f'(0)=f(x)\o x=x\o e_i\o x=e_i$. Thus, by hypothesis, $G(f')$ is a negative cycle. It follows from Lemma~\ref{lem:Hf} that $G(f)$ is a cycle, and since $f$ has no fixed points, $G(f)$ is a negative cycle. This proves {\bf(5)}.

\medskip
So according to {\bf(5)}, we can assume, without loss of generality, the following hypothesis:
\begin{enumerate}
\item[{\bf(H)}] $f(0)=e_i$ for some $i\in V$. 
\end{enumerate}
Also, in the all following, we use the following notations: 
\begin{enumerate}
\item[] 
$n=|V|$, ~$x^1=0$~ and ~$x^{k+1}=f(x^k)$~ for all $k\in\mathbb{N}$. 
\end{enumerate}
We first prove the following property (using arguments similar to the ones introduced in claim {\bf (1)} of the proof of Theorem~\ref{thm:circular-eosd}).
\begin{enumerate}
\item[{\bf (6)}]%6 
For all $k\geq 1$, there exists $i_k\in V$ such that $f(x^k)=x^k\o e_{i_k}$, and the resulting sequence $i_1i_2i_3\dots$ is a periodic sequence of period $n$. 
\end{enumerate}
We prove this by induction on $k$. The case $k=1$ is given by the the hypothesis {\bf(H)}, so suppose that $k>1$. Then $x^p=f(x^{p-1})=x^{p-1}\o e_{i_{p-1}}$ thus $d(x^p,x^{p-1})=1$, and since $f$ is non-expansive, we deduce that  
\[
d(f(x^p),x^p)=d(f(x^p),f(x^{p-1}))\leq d(x^p,x^{p-1})=1.
\]
Since $f$ has no fixed point $d(f(x^p),x^p)=1$ so there exists $i_k\in V$ such that $f(x^p)=x^p\o e_{i_p}$. We now prove that $i_1i_2i_3\dots$ is a periodic sequence of period~$n$. Let $k\geq 1$. Suppose that there exists $l\geq 1$ such that $i_k=i_{k+l}$, and let $l$ be minimal for this property. Then $\f(x^k)=\f(x^{k+l})=e_{i_k}$. Since 
\[
x^{k+l}=x^k\o e_{i_k}\o e_{i_{k+1}}\o\cdots \o e_{i_{k+l-1}}
\]
and since $i_k\neq i_{k+p}$ for all $1\leq p<l$ we have $x^{k+l}_{i_k}\neq x^k_{i_k}$ thus following {\bf (4)}, $x^{k+l}=x^k\o 1$. Consequently, $l=n$. Thus, the sequence $i_1i_2i_3\dots$ has period $n$ and {\bf (6)} is proved.  

\begin{enumerate}
\item[{\bf (7)}]%7 
As an immediate consequence of {\bf (6)}, we have 
\[
\begin{array}{lclcl}
x^1&=&0\\[1mm]
x^2&=&e_{i_1}\\[1mm]
x^3&=&e_{i_1}\o e_{i_2}\\[1mm]
&\vdots&\\[1mm]
x^k&=&e_{i_1}\o e_{i_2}\o e_{i_3}\o\cdots\o e_{i_{k-1}}\\[1mm]
&\vdots&\\[1mm]
x^{n+1}&=&e_{i_1}\o e_{i_2}\o e_{i_3}\o\cdots\o e_{i_{k-1}}\o\cdots\o e_{i_n}&=&1
\end{array}
\]
and
\[
\begin{array}{lclcl}
x^{n+1}&=&1\\[1mm]
x^{n+2}&=&1\o e_{i_1}\\[1mm]
x^{n+3}&=&1\o e_{i_1}\o e_{i_2}\\[1mm]
&\vdots&\\[1mm]
x^{n+k}&=&1\o e_{i_1}\o e_{i_2}\o e_{i_3}\o\cdots\o e_{i_{k-1}}\\[1mm]
&\vdots&\\[1mm]
x^{2n+1}&=&1\o e_{i_1}\o e_{i_2}\o e_{i_3}\o\cdots\o e_{i_{k-1}}\o\cdots\o e_{i_n}&=&0.
\end{array}
\]
\end{enumerate}

\medskip%
Let $h$ be the negative-circular network on $V$ such that $G(h)$ is the negative cycle with a negative arc from $i_{n}$ to $i_{n+1}=i_1$ and a positive arc from $i_{k}$ to $i_{k+1}$ for all $1\leq k<n$. In this way, for all $x\in\B^V$, 
\[
h_{i_1}(x)=x_{i_n}\o 1,\qquad h_{i_k}(x)=x_{i_{k-1}}\quad 1<k\leq n.
\]
We will prove that $h=f$, using several times the following easy tow next properties.

\begin{enumerate}%8
\item[{\bf (8)}] 
For all $x\in\B^V$ and $1\leq k<l\leq n$,
\[
\left.
\begin{array}{l}
f(x\o e_{i_k})=h(x\o e_{i_k})\\[1mm]
f(x\o e_{i_l})=h(x\o e_{i_l})
\end{array}
\right\}
\quad
\Rightarrow
\quad
\begin{array}{l}
f(x)=h(x)\text{ or}\\[1mm]
\qquad f(x)=h(x)\o e_{i_{k+1}}\o e_{i_{l+1}}.
\end{array}
\]
\end{enumerate} 
Since $f$ is non expansive, 
\[
\begin{array}{lllll}
d(f(x),f(x\o e_{i_k}))&=&d(f(x),h(x\o e_{i_k}))&\leq& 1\\[1mm]
d(f(x),f(x\o e_{i_l}))&=&d(f(x),h(x\o e_{i_l}))&\leq& 1\\[1mm]
\end{array} 
\]
Also $h(x\o e_{i_k})=h(x)\o e_{i_{k+1}}$ and 
$h(x\o e_{i_l})=h(x)\o e_{i_{l+1}}$. From $k\neq l$ it comes that $d(h(x\o e_{i_k}),h(x\o e_{i_l}))=2$ and thus  
\[
\begin{array}{lll}
d(f(x),h(x)\o e_{i_{k+1}})&=&1\\[1mm]
d(f(x),h(x)\o e_{i_{k+l}})&=&1
\end{array} 
\]
Hence, there exists $p,q$ such that 
\[
\begin{array}{lllll}
f(x)&=&h(x)\o e_{i_{k+1}}\o e_{i_p}\\[1mm]
f(x)&=&h(x)\o e_{i_{l+1}}\o e_{i_q}
\end{array} 
\]
Thus if $f(x)\neq h(x)$ then $i_{p}=i_{l+1}$ and $i_q=i_{k+1}$. This proves {\bf (8)}.  

\begin{enumerate}%9
\item[{\bf (9)}] 
For all $x\in\B^V$ and $1\leq k<l<p\leq n$, 
\[
\left.
\begin{array}{l}
f(x\o e_{i_k})=h(x\o e_{i_k})\\[1mm]
f(x\o e_{i_l})=h(x\o e_{i_l})\\[1mm]
f(x\o e_{i_p})=h(x\o e_{i_p})
\end{array}
\right\}
\quad
\Rightarrow
\quad
f(x)=h(x).
\]
\end{enumerate} 
Indeed, if $f(x)\neq h(x)$, then according to {\bf (8)},
\[
\begin{array}{lll} 
f(x)&=&h(x)\o e_{i_{k+1}}\o e_{i_{l+1}}\\[1mm]
f(x)&=&h(x)\o e_{i_{k+1}}\o e_{i_{p+1}}
\end{array}
\]
thus $i_{l+1}=i_{p+1}$, a contradiction. This proves {\bf (9)}. 

\begin{enumerate}
\item[{\bf (10)}]%10
If $x\in\B^V$ and $x_{i_1}>x_{i_n}$ then $f(x)=h(x)$.
\end{enumerate}
Let $x\in\B^V$ be such that $x_{i_1}=1$ and $x_{i_n}=0$. Consider the sequence $s(x)=x_{i_1}x_{i_2}\dots x_{i_n}$, and decompose this sequence into maximal subsequences with only $1$ or only $0$, in the following way:
\[
s(x)=\underbrace{11\cdots11}_{s(x)^1}\underbrace{00\cdots00}_{s(x)^2}\underbrace{11\cdots11}_{s(x)^3}\underbrace{00\cdots00}_{s(x)^4}11\cdots/\!/\cdots11\underbrace{00\cdots 00}_{s(x)^{t(x)}}.
\]
Clearly, $t(x)$ is even and $t(x)\geq 2$ (since $x_{i_1}>x_{i_n}$). For each $1\leq p\leq t(x)$, let $|s(x)^p|$ denote the length of $s(x)^p$.  

\begin{enumerate}
\item[{\bf (a)}]%A
Suppose that $t(x)=2$. Then $s(x)$ has the following form: 
\[
s(x)=x_{i_1}x_{i_2}\dots x_{i_n}=\underbrace{11\cdots11}_{s(x)^1}\underbrace{00\cdots00}_{s(x)^2}.
\]
Let $k$ be such that $x_{i_k}$ is the first element of $s(x)^2$ (or equivalently, the first zero of $s(x)$). Then $x=e_{i_1}\o e_{i_2}\o\cdots\o e_{i_{k-1}}$, so $h(x)=x\o e_{i_k}$. Following {\bf (7)} we have $x=x^k$ and $f(x^k)=x^k\o e_{i_k}$ thus $f(x)=h(x)$. 
\end{enumerate}
\begin{enumerate}
\item[{\bf (b)}]%B 
Suppose that $t(x)=4$. Then $s(x)$ has the following form: 
\[
s(x)=x_{i_1}x_{i_2}\dots x_{i_n}=\underbrace{11\cdots11}_{s(x)^1}\underbrace{00\cdots00}_{s(x)^2}
\underbrace{11\cdots11}_{s(x)^3}\underbrace{00\cdots00}_{s(x)^4}.
\]
We show that $f(x)=h(x)$ by induction on $|s(x)^2|$ and then on $|s(x)^3|$. Let $x_{i_k}$ be the first element of $s(x)^2$, let $x_{i_l}$ be the first element of~$s(x)^3$, and let $x_{i_p}$ be the last element of $s(x)^3$, so that:
\[
s(x)^2=x_{i_k}x_{i_{k+1}}\cdots x_{i_{l-1}}\qquad 
s(x)^3=x_{i_l}x_{i_{l+1}}\cdots x_{i_p}.
\] 
\begin{enumerate}
\item[$\bullet$]%B1
Suppose that $|s(x)^2|=1$. Assume first that $|s(x)^3|=1$. In this situation, $s(x)^2=x_{i_k}$, $s(x)^3=x_{i_{k+1}}$ and 
\[
x=e_{i_1}\o e_{i_2}\o \cdots \o e_{i_{k-1}}\o e_{i_{k+1}}
\]
so that 
\[
h(x)=x\o e_{i_k}\o e_{i_{k+1}}\o e_{i_{k+2}}.
\] 
Also $t(x\o e_{i_k})=2$ and $t(x\o e_{i_{k+1}})=2$, and from {\bf (a)} it follows that $f(x\o e_{i_k})=h(x\o e_{i_k})$ and $f(x\o e_{i_{k+1}})=h(x\o e_{i_{k+1}})$. Consequently, according to {\bf (8)}, we have $f(x)=h(x)$ or $f(x)=h(x)\o e_{i_{k+1}}\o e_{i_{k+2}}$. In the second case, 
\[
\begin{array}{rcl}
f(x)&=&h(x)\o e_{i_{k+1}}\o e_{i_{k+2}}\\[1mm]
&=&x\o e_{i_k}\o e_{i_{k+1}}\o e_{i_{k+2}}\o e_{i_{k+1}}\o e_{i_{k+2}}\\[1mm]
&=&x\o e_{i_k}.
\end{array}
\]
Thus $\f(x)=e_{i_k}$. Following {\bf (7)}, $\f(x^{n+k})=e_{i_k}$ and we deduce from {\bf (4)} that $x^{n+k}=x\o 1$, which is a contradiction since by,~{\bf (7)}, 
\[
x^{n+k}=1\o e_{i_1}\o e_{i_2}\o e_{i_3}\o\cdots\o e_{i_{k-1}}=1\o x\o e_{i_{k+1}}.
\]
Consequently, $f(x)=h(x)$. This proves the base case of {\bf (b1)}. For the induction step, assume that $|s(x)^3|>1$. Then $s(x)^2=x_{i_k}$, $s(x)^3=x_{i_{k+1}}\cdots x_{i_p}$ and 
\[
x=e_{i_1}\o e_{i_2}\o \cdots \o e_{i_{k-1}}\o e_{i_{k+1}}\o e_{i_{k+2}}\o\cdots \o e_{i_p}
\]
so that 
\[
h(x)=x\o e_{i_k}\o e_{i_{k+1}}\o e_{i_{p+1}}.
\]
Also $t(x\o e_{i_k})=2$ and we deduce from {\bf (a)} that $f(x\o e_{i_k})=h(x\o e_{i_k})$. In addition, $t(x\o e_{i_p})=4$ and 
\[
|s(x\o e_{i_p})^2|=1<|s(x\o e_{i_p})^3|=|s(x)^3|-1.
\]
Thus, by induction hypothesis, $f(x\o e_{i_p})=h(x\o e_{i_p})$. Hence, according to {\bf (8)} we have $f(x)=h(x)$ or $f(x)=h(x)\o e_{i_{k+1}}\o e_{i_{p+1}}$. In the second case, 
\[
\begin{array}{rcl}
f(x)&=&h(x)\o e_{i_{k+1}}\o e_{i_{p+2}}\\[1mm]
&=&x\o e_{i_k}\o e_{i_{k+1}}\o e_{i_{p+1}}\o e_{i_{k+1}}\o e_{i_{p+1}}\\[1mm]
&=&x\o e_{i_k}.
\end{array}
\]
Thus $\f(x)=e_{i_k}$. Following {\bf (7)}, $\f(x^{n+k})=e_{i_k}$ and we deduce from {\bf (4)} that $x^{n+k}=x\o 1$, which is a contradiction since, by~{\bf (7)}, 
\[
\begin{array}{rcl}
x^{n+k}&=&1\o e_{i_1}\o e_{i_2}\o e_{i_3}\o\cdots\o e_{i_{k-1}}\\[1mm]
&=&1\o x\o e_{i_{k+1}}\o e_{i_{k+2}}\o\cdots \o e_{i_p}.
\end{array}
\]
Consequently, $f(x)=h(x)$.
\item[$\bullet$]%B2
Suppose that $|s(x)^2|>1$. Then $t(x\o e_{i_k})=t(x\o e_{i_{l-1}})=4$, and $|s(x\o e_{i_k})^2|=|s(x\o e_{i_{l-1}})^2|<|s(x)^2|$. Thus, by induction hypothesis, 
\[
\begin{array}{lcl}
f(x\o e_{i_k})&=&h(x\o e_{i_k})\\[1mm]
f(x\o e_{i_{l-1}})&=&h(x\o e_{i_{l-1}})
\end{array}
\]
Suppose that $|s(x)^3|=1$ so that $s(x)^3=x_{i_l}$. Then $t(x\o e_{i_l})=2$ and we deduce from {\bf (a)} that $f(x\o e_{i_l})=h(x\o e_{i_l})$ and from {\bf (9)} it comes that $f(x)=h(x)$. Now, suppose that $|s(x)^3|>1$. Then $t(x\o e_{i_p})=4$, $|s(x\o e_{i_p})^2|=|s(x)^2|$ and $|s(x\o e_{i_p})^3|=|s(x)^3|-1$, thus, by induction hypothesis, 
\[
\begin{array}{lcl}
f(x\o e_{i_p})&=&h(x\o e_{i_p})
\end{array}
\]
and according to {\bf (9)}, $f(x)=h(x)$. 
\end{enumerate}
\end{enumerate}

\begin{enumerate}
\item[{\bf (c)}]%C
Suppose that $t(x)\geq 4$. We prove that $f(x)=h(x)$ by induction on $t(x)$ and then on $|s(x)^2|+|s(x)^4|$. The base case $t(x)=4$ is given by {\bf (b)}. So assume that $t(x)\geq 6$. We use the following notations:
\[
\begin{array}{rcl}
s(x)^2&=&x_{i_k}x_{i_{k+1}}\cdots x_{i_{q}}\\[1mm]
s(x)^4&=&x_{i_l}x_{i_{l+1}}\cdots x_{i_p}\\[1mm]
s(x)^{t(x)-1}&=&x_{i_r}x_{i_{r+1}}\cdots x_{i_s}\\[1mm]
\end{array}
\] 
\begin{enumerate}
\item[$\bullet$]%C1
Suppose that $|s(x)^2|+|s(x)^4|=2$. Then $t(x\o e_{i_k})=t(x\o e_{i_l})=t(x)-2$. Thus, by induction hypothesis.  
\[
\begin{array}{lcl}
f(x\o e_{i_k})&=&h(x\o e_{i_k})\\[1mm]
f(x\o e_{i_l})&=&h(x\o e_{i_l})
\end{array}
\]
We prove that $f(x)=h(x)$ by induction on $|s(x)^{t(x)-1}|$. If $|s(x)^{t(x)-1}|=1$ then $t(x\o e_{i_r})=t(x)-2$ and by induction hypothesis, 
\[
\begin{array}{lcl}
f(x\o e_{i_r})&=&h(x\o e_{i_r}).
\end{array}
\]
Thus according to {\bf (9)}, $f(x)=h(x)$. If $|s(x)^{t(x)-1}|>1$ then $t(x\o e_{i_s})=t(x)$, $|s(x\o e_{i_s})^2|+|s(x\o e_{i_s})^4|=2$ and $|s(x\o e_{i_s})^{t(x)-1}|<|s(x)^{t(x)-1}|$, thus, by induction hypothesis, 
\[
\begin{array}{lcl}
f(x\o e_{i_s})&=&h(x\o e_{i_s})
\end{array}
\]
and according to {\bf (9)}, $f(x)=h(x)$. 
\item[$\bullet$]%C2
Suppose that $|s(x)^2|+|s(x)^4|>2$. Then either $|s(x)^2|\geq 2$ or $|s(x)^4|\geq 2$. Suppose that $|s(x)^2|\geq 2$, the other case being similar. Then $t(x\o e_{i_k})=t(x\o e_{i_q})=t(x)$ and $|s(x\o e_{i_k})^2|=|s(x\o e_{i_q})^2|<|s(x)^2|$ and $|s(x\o e_{i_k})^4|=|s(x\o e_{i_q})^4|=|s(x)^4|$, and so, by induction hypothesis, 
\[
\begin{array}{lcl}
f(x\o e_{i_k})&=&h(x\o e_{i_k})\\[1mm]
f(x\o e_{i_q})&=&h(x\o e_{i_q}).
\end{array}
\]
If $|s(x)^4|=1$ then $t(x\o e_{i_l})=t(x)-2$ thus, by induction hypothesis, $f(x\o e_{i_l})=h(x\o e_{i_l})$; otherwise, $t(x\o e_{i_l})=t(x)$ and $|s(x\o e_{i_l})^2|=|s(x)^2|$ and $|s(x\o e_{i_l})^4|<|s(x)^4|$, and so, by induction hypothesis, we have again  
\[
\begin{array}{lcl}
f(x\o e_{i_l})&=&h(x\o e_{i_l}).
\end{array}
\]
Thus, according to {\bf (9)}, $f(x)=h(x)$. This ends the proof of {\bf (10)}.\end{enumerate}
\end{enumerate}

\medskip
With similar arguments, we get:
\begin{enumerate}
\item[{\bf (11)}]%6
If $x\in\B^V$ and $x_{i_1}<x_{i_n}$ then $f(x)=h(x)$.
\end{enumerate}

\medskip
Hence, to complete the proof, it remains to prove that if $x_{i_1}=x_{i_n}$ then $f(x)=h(x)$. Assume that $x_{i_1}=x_{i_n}=0$. We proceed by induction on $\d{x}$. If $\d{x}=0$ then $f(x)=h(x)$ according to {\bf (7)}. Otherwise, there exists $1<k<n$ such that $x_{i_k}=1$. Since $\d{x\o e_{i_k}}=\d{x}-1$, by induction hypothesis,
\[
\begin{array}{lcl}
f(x\o e_{i_k})&=&h(x\o e_{i_k}).
\end{array}
\]  
Now since $(x\o e_{i_1})_{i_1}>(x\o e_{i_1})_{i_n}$ and $(x\o e_{i_n})_{i_1}<(x\o e_{i_n})_{i_n}$, according to {\bf (10)} and {\bf (11)} we have
\[
\begin{array}{lcl}
f(x\o e_{i_1})&=&h(x\o e_{i_1})\\[1mm]
f(x\o e_{i_n})&=&h(x\o e_{i_n})
\end{array}
\]
and we deduce from {\bf (9)} that $f(x)=h(x)$. If $x_{i_1}=x_{i_n}=1$, we prove with similar arguments that $f(x)=h(x)$. Thus $f=h$.
\end{proof}
%%%%%%%%%%%%%%%%%%

As a consequence of this theorem and the fact that a network with multiple fixed points (resp. without fixed point) has a 2-critical (resp. 0-critical) subnetwork, we obtain the following ``dichotomization'' of Corollary~\ref{cor:nonexp1}.   

%%%%%%%%%%%%%%%%%%
\begin{corollary} Suppose that $f$ is non-expansive.
\begin{enumerate}
\item 
Each subnetwork of $f$ has at most one fixed point if and only if $f$ has no positive-circular subnetwork. 
\item 
Each subnetwork of $f$ has at least one fixed point if and only if $f$ has no negative-circular subnetwork.  
\end{enumerate}
\end{corollary}
%%%%%%%%%%%%%%%%%%

It is easy to see that, if the maximal in-degree of the global interaction graph $G(f)$ of a network $f$ is at most one, then $Gf(x)=G(f)$ for all $x\in\B^V$. Thus, in particular, if $f$ is circular then $Gf(x)=G(f)$ for all $x\in\B^V$. Proceeding as in Section~\ref{sec:shih} with this property instead of  Proposition~\ref{odd}, we obtain the following corollary. Note that the second point generalizes Theorem~\ref{thm:non-exp}. 

%%%%%%%%%%%%%%%%%%
\begin{corollary} Suppose that $f$ is non-expansive.
\begin{enumerate}
\item 
If, for every $1\leq k\leq |V|$, there exists at most $2^k-1$ points $x$ such that $Gf(x)$ has a chordless positive cycle of length $k$, then $f$ has at most one fixed points. 
\item 
If, for every $1\leq k\leq |V|$, there exists at most $2^k-1$ points $x$ such that $Gf(x)$ has a chordless negative cycle of length $k$, then $f$ has at least one fixed points. 
\end{enumerate}
\end{corollary}
%%%%%%%%%%%%%%%%%%

%%%%%%%%%%%%%%%%%%%%%%%%%%%%%%%%%%%%%%%%%%%%%%%%%%%%%%%%%%%%%%%%%%%%%%
%%%%%%%%%%%%%%%%%%%%%%%%%%%%%%%%%%%%%%%%%%%%%%%%%%%%%%%%%%%%%%%%%%%%%%
\section{Conjonctive networks}\label{sec:and}
%%%%%%%%%%%%%%%%%%%%%%%%%%%%%%%%%%%%%%%%%%%%%%%%%%%%%%%%%%%%%%%%%%%%%%
%%%%%%%%%%%%%%%%%%%%%%%%%%%%%%%%%%%%%%%%%%%%%%%%%%%%%%%%%%%%%%%%%%%%%%

A network $f$ on $V$ is an {\bf and-net} (or {\bf conjunctive network}) if $G(f)$ is simple and if, for every $i\in V$, $f_i$ is the conjunction of the positive and negative inputs of $i$ in $G(f)$, that is: For all $x\in \B^V$, $f_i(x)=1$ if and only if $G(f)$ has no positive arc $j\to i$ with $x_j=0$ and no negative arc $j\to i$ with $x_j=1$. Note that every subnetwork of an and-net is an and-net. Note also that for the class of and-nets, $f$ and $G(f)$ share the same informations.

\medskip
In this section, we first prove that every 2-critical and-net is positive circular (but we were not able to prove that every 0-critical and-net is negative circular). Then, we show that, for and-nets, the presence of even-self-dual (resp. odd-self-dual) subnetworks can be checked in a very simple way by looking at the chordless positive (resp. negative) cycles of $G(f)$.

%%%%%%%%%%%%%%%%%%
\begin{proposition}\label{pro:and1} 
~
\begin{enumerate}
\item
$f$ if positive-circular if and only if $f$ is an even-self-dual and-net.
\item
$f$ if negative-circular if and only if $f$ is an odd-self-dual and-net.
\end{enumerate}
\end{proposition}
%%%%%%%%%%%%%%%%%%

%%%%%%%%%%%%%%%%%%
\begin{proof}
Suppose that $f$ is positive-circular (negative-circular). Then, by Theorem~\ref{thm:circular-eosd}, $f$ is even-self-dual (resp. odd-self-dual), and since each vertex $i\in V$ has exactly one in-neighbor in $G(f)$, $f$ is an and-net. 

\medskip
Suppose that $f$ is an even- or odd-self-dual and-net. Let $i,j,k\in V$, and assume that $j$ and $k$ are distinct in-neighbor of $i$ in $G(f)$. Let $x\in\B^V$ be such that $f_i(x)=1$. Then $f_i(y)=0$ for every $y$ such that $y_j\neq x_j$ or $y_k\neq x_k$. So $f_i(x\o e_j)=f_i(x \o e_j\o 1)=0$, so $f$ is not self-dual, a contradiction. Since $f_i$ is not a constant, we deduce each vertex of $G(f)$ is of in-degree one. According to Proposition~\ref{odd}, each vertex of $G(f)$ is of out-degree at least one. Since the sum of the in-degrees equals the sum of the out-degrees, we deduce that each vertex of $G(f)$ is of in-degree one and out-degree one~\footnote{If each vertex of $G(f)$ is of out-degree one, then $f$ is non-expansive, and we can conclude by applying Theorem~\ref{thm:circular-eosd}. However, we give here the few additional arguments that makes the proof independent of Theorem~\ref{thm:circular-eosd}.}. In other words, $G(f)$ is a disjoint union of cycles. Let $C$ be a cycle of $G(f)$ with vertex set $I$. Then, for all $x\in\B^V$, 
\[ 
f_i(x\o e_I)= f_i(x)\o 1\qquad \forall i\in I,
\]
and since $G(f)$ has no arc from $I$ to $V\setminus I$, we deduce that 
\[
f_i(x\o e_I)=f_i(x)\qquad \forall i\in V\setminus I,
\]
So $f(x\o {e_I})=f(x)\o e_I$, and thus:
\[
\f(x\o e_I)=(x\o e_I)\o (f(x)\o e_I)=x\o f(x)=\f(x).
\]
and since $f$ is even- or odd-self-dual, we deduce that $x\o e_I=x\o 1$, that is $I=V$. So $G(f)$ is a cycle, which is positive if $f$ is even, and negative otherwise.
\end{proof}
%%%%%%%%%%%%%%%%%%

Using this proposition and Corollary~\ref{cor2} we obtain the following characterization. 

%%%%%%%%%%%%%%%%%%
\begin{corollary}\label{cor:and1} 
If $f$ is an and-net, then each subnetwork of $f$ has a unique fixed point if and only if $f$ has no circular subnetworks.
\end{corollary}
%%%%%%%%%%%%%%%%%%

We will now show that the ``unicity part'' of this characterization can be obtained under the absence of positive-circular subnetwork.  

%%%%%%%%%%%%%%%%%%
\begin{theorem}\label{thm:and2}
$f$ is positive-circular if and only if $f$ is a 2-critical and-net.  
\end{theorem}
%%%%%%%%%%%%%%%%%%

%%%%%%%%%%%%%%%%%%
\begin{proof}
If $P$ is a sequence of signed arcs of $G(f)$, we set $s(P)=0$ if $P$ has an even number of negative arcs, and $s(P)=1$ if $P$ has an odd number of negative arcs. We first prove the following two properties (which may be of independent interest). 
\begin{enumerate}
\item[{\bf(1)}]
Suppose that $f$ is an and-net. Suppose also that there exists $x\in\B^V$ such that $f(x)=x$ and $f(x\o 1)=x\o 1$. Let  
\[
P=(i_1,s_1,i_2),(i_2,s_2,i_3),\dots,(i_{l-1},s_{l-1},i_{l}),(i_l,s_l,i_{l+1})
\]
be a sequence of arcs of $G(f)$. Then $s(P)=x_{i_1}\o x_{i_{l+1}}$. 
\end{enumerate}
We proceed by induction of the length $l$ of the sequence. 
\begin{enumerate}
\item
Suppose that $l=1$, that is $P=(i_1,s_1,i_2)$. If $s_1=1$, then the arc from $i_1$ to $i_2$ is positive and: If  $x_{i_1}=0$ then $f_{i_2}(x)=0=x_{i_2}$; and if $x_{i_1}=1$, then $f_{i_2}(x\o 1)=0=x_{i_2}\o 1$ thus $x_{i_2}=1$. Hence, in both cases, $x_{i_1}\o x_{i_2}=0=s(P)$. If $s_1=-1$, then the arc from $i_1$ to $i_2$ is negative so: If $x_{i_1}=1$ then $f_{i_2}(x)=0=x_{i_2}$; and if $x_{i_1}=0$, then $f_{i_2}(x\o 1)=0=x_{i_2}\o 1$ thus $x_{i_2}=1$. Hence, in both cases, $x_{i_1}\o x_{i_2}=1=s(P)$. This prove the base case.
\item
Suppose that $l>1$. Then $P$ can be expressed as the concatenation $P=QQ'$ of two subsequences $Q$ and $Q'$, both of length at most $l-1$. If $q$ is the length of $Q$, then, by induction hypothesis, $s(Q)=x_{i_1}\o x_{i_{q+1}}$ and  $s(Q')=x_{i_{q+1}}\o x_{i_{l+1}}$ thus $s(P)=s(Q)\o s(Q')=x_{i_1}\o x_{i_{q+1}}\o x_{i_{q+1}}\o x_{i_{l+1}}=x_{i_1}\o x_{i_{l+1}}$. This proves {\bf(1)}.
\end{enumerate}
\begin{enumerate}
\item[{\bf(2)}]
Suppose that $f$ is an and-net. If there exists $x\in\B^V$ such that $f(x)=x$ and $f(x\o 1)=x\o 1$ then $G(f)$ has no negative cycle. 
\end{enumerate}
If $C$ is a cycle of $G(f)$ go length $l$, and if $P=(i_1,s_1,i_2),(i_2,s_2,i_3),\dots,(i_l,s_l,i_1)$ are the arcs of $C$ given in the order, then following {\bf(1)}, $s(P)=x_{i_1}\o x_{i_1}=0$, thus $C$ has an even number of negative arcs, {\em i.e.} $C$ is positive. This proves~{\bf(2)}.

\medskip
We are now in position to prove the theorem. By Theorem~\ref{thm:nonexp}, every positive-circular network is 2-critical, and it is obvious that positive-circular networks are and-nets. So assume that $f$ is a 2-critical and-net. By theorem~\ref{main} and Proposition~\ref{pro:and1}, $f$ has a positive- or negative-circular subnetwork $h$. Following {\bf(2)}, $h$ cannot be negative-circular. Thus $h$ is positive-circular. Thus $h$ has two fixed points, and since $f$ is 2-critical, $h=f$.\end{proof} 
%%%%%%%%%%%%%%%%%%

As a consequence of this theorem and the fact that a network with multiple fixed points has a 2-critical subnetwork, we obtain the following characterization. 

%%%%%%%%%%%%%%%%%%
\begin{corollary}\label{cor:and2} 
If $f$ is an and-net, then each subnetwork of $f$ has at most one fixed point if and only if $f$ has no positive-circular subnetworks.
\end{corollary}
%%%%%%%%%%%%%%%%%%

Using again the fact that if $f$ is circular then $Gf(x)=G(f)$ for all $x\in\B^V$, we obtain: 

%%%%%%%%%%%%%%%%%%
\begin{corollary} Suppose that $f$ is an and-net. 
If, for every $1\leq k\leq |V|$ there exists at most $2^k-1$ points $x$ such that $Gf(x)$ has a chordless positive cycle of length $k$, then $f$ has at most one fixed points. 
\end{corollary}
%%%%%%%%%%%%%%%%%%

%%%%%%%%%%%%%%%%%%
\begin{rem}
In view of Theorems~\ref{thm:nonexp} and \ref{thm:and2}, it is tempting to think that {\em every 0-critical and-net is negative-circular}. But this is false, as showed below. For all $n\geq 4$, let $G_n$ be the digraph with vertex set $V=\{0,1,\dots,n-1\}$ and such that for all $u\in V$ and $k\in\{1,\pm2,\pm3,\dots,\pm\lfloor\frac{1}{2}n\rfloor\}$ there is an arc from $u$ to $u+k~(\mathrm{mod}~n)$. In \cite{GN86}, it is proved that $G_n$ is {\em kernel-critical}: $G_n$ has no kernel and every strict induced subdigraph has a kernel. Using the correspondence between kernels in digraphs and fixed points in and-nets established in \cite{RR2012}, we easily deduce that: {\em For all $n\geq 4$, the and-net $f$ such that $|G(f)|=G_n$ and such that $G(f)$ has only negative arc is a  non-circular $0$-critical and-net.} 
\end{rem}
%%%%%%%%%%%%%%%%%%

Now, we show how to check if an and-net has or not a circular subnetworks by looking at the chordless cycles of $G(f)$. For that, additional definitions are needed. Let $G$ be a simple interaction graph with vertex set $V$, and let $C$ be a cycle in it. A vertex $v\in V$ is a {\bf delocalizing vertex} of $C$ if $G$ has both a positive and a negative arcs from $v$ to distinct vertices of $C$ ($v$ can be a vertex of $C$; in such a case the cycle has two chords of opposite sign starting from $v$).

%%%%%%%%%%%%%%%%%%
\begin{proposition}[Richard and Ruet \cite{RR2012}]\label{pro:and3}
Suppose that $f$ is an and-net. There exists $x\in\B^V$ such that $Gf(x)$ has a cycle $C$ if and only $C$ is a cycle of $G(f)$ that has no delocalizing vertex in $G(f)$. 
\end{proposition}
%%%%%%%%%%%%%%%%%%

%%%%%%%%%%%%%%%%%%
\begin{proposition}[Remy and Ruet \cite{RR08}]\label{pro:and4}
Let $f$ be a network on $V$. Let $x\in\B^V$, and suppose that $Gf(x)$ has a cycle $C$ with vertex set $I$. If $C$ has no chord in $G(f)$, then the subnetwork of $f$ induced by $x_{-I}$ is a circular network with interaction graph $C$. 
\end{proposition}
%%%%%%%%%%%%%%%%%%

%%%%%%%%%%%%%%%%%%
\begin{proposition}\label{pro:and5}
Suppose that $f$ is an and-net. Then $f$ has a circular subnetwork with interaction graph $C$ if and only if $C$ is a cycle of $G(f)$ that has no chord and no delocalizing vertex in $G(f)$.
\end{proposition}
%%%%%%%%%%%%%%%%%%

%%%%%%%%%%%%%%%%%%
\begin{proof} 
If $C$ a cycle of $G(f)$ without chord and delocalizing vertex in $G(f)$, then the fact that $f$ has a subnetwork with interaction graph $C$ follows from Proposition~\ref{pro:and3} and Proposition~\ref{pro:and4}. 

\medskip
Suppose that $h$ is a circular subnetwork of $f$ with interaction graph~$C$. Let $I$ be the vertex set of $C$, $x\in\B^V$, and suppose that $h$ is induced by $x_{-I}$. Since $Gh(x_{-I})=G(h)=C$ is a subgraph of $Gf(x)$, we deduce from Proposition~\ref{pro:and3} that $C$ has no delocalizing vertex in $G(f)$. Suppose, for a contradiction, that $C$ has a chord in $G(f)$, say from $j$ to $i$. Let $k\neq j$ be the vertex preceding $i$ in $C$. 
Let $y\in\B^V$  be such that $y_{-I}=x_{-I}$ and $y_j=0$ if and only if the chord $j\to i$ is positive. Then $h_i(y_{-I})=f_i(y)=0$ and $h_i(y_{-I}\o e_k)=f_i(y\o e_k)=0$, thus $Gh(y_{-I})$ has no arc from $k$ to $i$, a contradiction with the fact that $h$ is circular. 
\end{proof}
%%%%%%%%%%%%%%%%%%

We are now in position to express conditions in Corollaries \ref{cor:and1} and \ref{cor:and2} in terms of chordless cycles and delocalizing vertices. 

%%%%%%%%%%%%%%%%%%
\begin{corollary}\label{theo:and2}
Let $f$ be an and-net.
\begin{enumerate}
\item
Each subnetwork of $f$ has a unique fixed point if and only if $f$ every chordless cycle of $G(f)$ has a delocalizing vertex. 
\item
Each subnetwork of $f$ has at most one fixed point if and only if $f$ every chordless positive cycle of $G(f)$ has a delocalizing vertex. 
\end{enumerate}
\end{corollary}
%%%%%%%%%%%%%%%%%%

%%%%%%%%%%%%%%%%%%
\begin{remark}
If $G(f)$ has $n$ vertices and $c$ cycles, then the enumeration of these $m$ cycles can be done with time complexity $\mathcal{O}(n^2c)$; see {\cite{J75}} for instance. Since for each cycle the absence chord and delocalizing vertex can be verified in $\mathcal{O}(n^2)$, conditions of Corollary~\ref{theo:and2} can be verified in time $\mathcal{O}(n^2c)$.
\end{remark}
%%%%%%%%%%%%%%%%%%

%%%%%%%%%%%%%%%%%%
\begin{example}[Continuation of Example \ref{ex1}] 
Take again the network $f$ on $\{1,2,3\}$ defined by:
\[
\begin{array}{l}
f_1(x)=\overline{x_2}\land x_3\\[1mm]
f_2(x)=\overline{x_3}\land x_1\\[1mm]
f_3(x)=\overline{x_1}\land x_2.
\end{array}
\]
This network is an and-net and its global interaction graph $G(f)$ is 
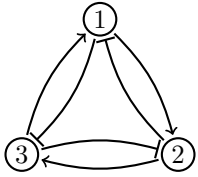
\[
\begin{tikzpicture} 
\pgfmathparse{1.9}
\node[outer sep=1.5,inner sep=1.5,circle,draw,thick] (1) at (90:1.2){$1$};
\node[outer sep=1.5,inner sep=1.5,circle,draw,thick] (2) at (-30:1.2){$2$};
\node[outer sep=1.5,inner sep=1.5,circle,draw,thick] (3) at (210:1.2){$3$};
\path[thick]
(1) edge[bend left=15,->] (2)
(2) edge[bend left=15,->] (3)
(3) edge[bend left=15,->] (1)
(1) edge[bend left=15,-|] (3)
(3) edge[bend left=15,-|] (2)
(2) edge[bend left=15,-|] (1)
;
\end{tikzpicture}
\]
It is easy to see that every chordless cycle ({\em i.e.} cycle of length 2) has a delocalizing vertex. Thus $f$ has no circular subnetwork (cf. Proposition~\ref{pro:and5}). Thus it has no even- or odd-self-dual subnetwork (cf. Proposition~\ref{pro:and1}). Thus each subnetwork of $f$ has a unique fixed point (cf. Corollary~\ref{cor2}); see indeed Example~\ref{ex1}. Note that the two cycles of length three have no delocalizing vertex, thus these cycles are in $Gf(x)$ for some $x$; see indeed Example~\ref{ex2}. 
\end{example}
%%%%%%%%%%%%%%%%%%

%%%%%%%%%%%%%%%%%%
\paragraph{Acknowledgements}
%%%%%%%%%%%%%%%%%%

I wish to thank Emmanuelle Seguin and Sebastien Brun for their help. This work has been partially supported by the French National Agency for Research (ANR-10-BLANC-0218 BioTempo project).

%%%%%%%%%%%%%%%%%%%%%%%%%%%%%%%%%%%%%%%%%%%%%%%%%%%%%%%%%%%%%%%%%%%%%%
%%%%%%%%%%%%%%%%%%%%%%%%%%%%%%%%%%%%%%%%%%%%%%%%%%%%%%%%%%%%%%%%%%%%%%
\appendix
\section{Proofs of Theorems \ref{thm:robert}, \ref{thm:thomas} and \ref{thm:shihdong}}
\label{appendix}
%%%%%%%%%%%%%%%%%%%%%%%%%%%%%%%%%%%%%%%%%%%%%%%%%%%%%%%%%%%%%%%%%%%%%%
%%%%%%%%%%%%%%%%%%%%%%%%%%%%%%%%%%%%%%%%%%%%%%%%%%%%%%%%%%%%%%%%%%%%%%

%%%%%%%%%%%%%%%%%%
\begin{proof}[Proof of Theorem~\ref{thm:robert}]
We proceed by induction on $|V|$. If $|V|=1$ the theorem is obvious. So assume that $|V|>1$, and suppose that $G(f)$ has no cycle. Then $G(f)$ has at least one vertex, say $i$, without in-neighbor. Hence, $f_i=\mathrm{cst}=\a\in\B$. Since $G(f^{i\a})$ is a subgraph of $G(f)$, $G(f^{i\a})$ has no cycle and by induction hypothesis $f^{i\a}$ has a unique fixed point {\em i.e.} there exists a unique $x\in\B^V$ with $x_i=\alpha$ such that $f^{i\a}(x_{-i})=x_{-i}$. Since $f(x)_{-i}=f^{i\a}(x_{-i})=x_{-i}$ and $f_i(x)=\a=x_i$, we deduce that $f(x)=x$. Suppose that $f$ has a fixed point $y\neq x$. Then $f_i(y)=\a=y_i$ so $y_{-i}\neq x_{-i}$ and $f^{i\a}(y_{-i})=f(y)_{-i}=y_{-i}$. Thus $f^{i\a}$ has a fixed point distinct from $x_{-i}$, a contradiction. Thus $x$ is the unique fixed point of $f$.
\end{proof}
%%%%%%%%%%%%%%%%%%

%%%%%%%%%%%%%%%%%%
\begin{proof}[Proof of Theorem~\ref{thm:thomas}]
We proceed by induction on the number of strongly connected components. If $|V|=1$ then the theorem is obvious. So assume that $|V|>1$. If $G(f)$ is strongly connected, then the theorem is given by Theorem~\ref{thm:ara}. So suppose that $G(f)$ is not strongly connected, an let $I\subseteq V$ be an initial strongly connected component of $G(f)$ (there is no arc from $V\setminus I$ to $I$). Let $h$ be the subnetwork of $f$ induced by $0\in\B^{V\setminus U}$. Let us prove that
\[\tag{$*$}
\forall x\in\B^V,\qquad h(x|_I)=f(x)|_I.
\]
Suppose, for a contradiction, that $h(x|_I)\neq f(x)|_I$ for some $x\in\B^V$, and assume that $\d{x}$ is minimal for this property. Then, since $h$ is induced by the point $0\in\B^{V\setminus I}$, there exists $j\in V\setminus I$ with $x_j=1$. Thus $\d{x\o e_j}<\d{x}$ so $h(x|_I)=h((x\o e_j)|_I)=f(x\o e_j)|_I\neq f(x)|_I$. Thus there exists $i\in I$ such that $f_i(x\o e_j)\neq f_i(x)$. Thus $G(f)$ has an arc from $j$ to $i$, a contradiction. This prove ($*$). We are now in position to complet the induction step. 
\begin{enumerate}
\item
Suppose that $G(f)$ has no positive cycle, and suppose, for a contradiction, that $x$ and $y$ are fixed points of $f$. Then following ($*$), $x|_I$ and $y|_I$ are fixed points of $h$. Since  $G(h)$ has no positive cycle, by induction hypothesis, $h$ has at most one fixed point, thus $x|_I=y|_I=z$. Let $h'$ be the subnetwork of $f$ induced by~$z$. By definition, $h'(x_{-I})=f(x)_{-I}$ and $h'(y_{-I})=f(y)_{-I}$. Thus $x_{-I}$ and $y_{-I}$ are fixed points of $h'$. Since $G(h')$ has no positive cycle, by induction hypothesis, $h'$ has at most one fixed point, thus $x_{-I}=y_{-I}$. Thus $x=y$ so $f$ has at most one fixed point. 
\item
Suppose that $G(f)$ has no negative cycle. Then $G(h)$ has no negative cycle, and by induction hypothesis, $h$ has at least one fixed point $z\in\B^I$. Let $h'$ be the subnetwork of $f$ induced by $z$. Again, by induction hypothesis, $h'$ has at least one fixed point. Thus, there exists $x\in\B^V$ with $x|_I=z$ such that $x_{-I}=h'(x_{-I})=f(x)_{-I}$, and by ($*$) we have $x|_I=z=h(z)=h(x|_I)=f(x)|_I$. Thus $x$ is a fixed point~of~$f$. 
\end{enumerate}    
\end{proof}
%%%%%%%%%%%%%%%%%%

%%%%%%%%%%%%%%%%%%
\begin{proof}[Proof of Theorem~\ref{thm:shihdong}]
The ``trick'' consists in proving, by induction on $|V|$, the following
more general statement:
\begin{enumerate}
\item[($*$)] If $Gf(x)$ has no cycle for all $x\in\B^V$, then the conjugate of $f$ is a
bijection (and so $f$ has a unique fixed point).
\end{enumerate}
The case $|V|=1$ is obvious. So suppose that $|V|>1$, and suppose that $Gf(x)$ has no cycle for all $x\in\B^V$. Let $i\in V$ and $\a\in\B$. For all $x\in\B^V$, $Gf^{i\a}(x_{-i})$ is a subgraph of $Gf(x)$, and thus $Gf^{i\a}(x_{-i})$ has no cycle. Using the induction hypothesis, we deduce that: {\em For all $i\in V$ and $\a\in\B$, the conjugate of $f^{i\a}$ is a bijection.}  Now, suppose that $\f$ is not a bijection. Then, there exists two distinct points $x$ and $y$ in $\B^V$ such that $\f(x)=\f(y)$. Let us proved that $x=y\o 1$. Indeed, if $x_i=y_i=\a$ for some $i\in V$, then $\f^{i\a}(x_{-i})=\f(x)_{-i}=\f(y)_{-i}=\f^{i\a}(y_{-i})$. Thus the conjugate of  $f^{i\a}$ is not a bijection, a contradiction. So $x=y\o 1$. Since $Gf(x)$ has no cycle, it contains at least one vertex of out-degree $0$. In other words, there exists $i\in V$ such that $f(x^{i1})=f(x^{i0})$. Thus $\f(x^{i1})_{-i}=\f(x^{i0})_{-i}=\f(x)_{-i}$. Hence, setting $\a=y_i$, we obtain
\[
\f^{i\a}(x_{-i})
=\f(x^{i\a})_{-i}
=\f(x)_{-i}
=\f(y)_{-i}
=\f(y^{i\a})_{-i}
=\f^{i\a}(y_{-i}).
\]
So the conjugate of $f^{i\a}$ is not a bijection, a contradiction. Thus $\f$ is a bijection and ($*$) is proved.
\end{proof}
%%%%%%%%%%%%%%%%%%

%%%%%%%%%%%%%%%%%%%%%%%%%%%%%%%%%%%%%%%%%%%%%%%%%%%%%%%%%%%%%%%%%%%%%%
%%%%%%%%%%%%%%%%%%%%%%%%%%%%%%%%%%%%%%%%%%%%%%%%%%%%%%%%%%%%%%%%%%%%%%
\bibliographystyle{plain}
\bibliography{bibliography}
\end{document}